\documentclass{lmcs} 

\usepackage[english]{babel}
\usepackage[utf8]{inputenc}
\usepackage{tikz}
\usetikzlibrary{positioning,calc,arrows,shapes,  tikzmark, decorations.pathreplacing}
\usetikzlibrary{automata}

\usepackage{amsmath}
\usepackage{graphicx}
\usepackage[colorlinks=true, allcolors=blue]{hyperref}
\usepackage{virginialake}
\usepackage[textsize=footnotesize, disable]{todonotes}
\usepackage{cleveref}
\usepackage{tikz-cd}

\usepackage[shortlabels]{enumitem}
\renewcommand{\emptyset}{\varnothing}
\renewcommand{\phi}{\varphi}
\newcommand{\dfn}{:=}

\renewcommand{\epsilon}{\varepsilon}

\newcommand{\red}[1]{{\color{red}#1}}
\newcommand{\blue}[1]{{\color{blue}#1}}
\newcommand{\orange}[1]{{\color{orange}#1}}

\definecolor{mygreen}{rgb}{0, 0.5, 0}

\definecolor{airforceblue}{rgb}{0.36, 0.54, 0.66}

\newcommand{\todoinline}[1]{\todo[caption={}, inline]{#1}}

\newcommand{\lukasedit}[2]{\color{red}{#2}\color{black}}

\newcommand{\Bool}{\{0,1\}}
\newcommand{\Nat}{\mathbb N}

\newcommand{\Pij}[2]{\Pi^{#1}_{#2}}
\newcommand{\Sij}[2]{\Sigma^{#1}_{#2}}
\newcommand{\Dij}[2]{\Delta^{#1}_{#2}}

\newcommand{\lfp}[1]{\mu#1}
\newcommand{\gfp}[1]{\nu#1}
\renewcommand{\succ}{\mathsf s}
\newcommand{\pow}[1]{\mathcal P (#1)}

\newcommand{\formop}[3]{ #1_{ #3, #2, \MakeUppercase{#2}} }

\newcommand{\code}[1]{\ulcorner #1 \urcorner}
\newcommand{\sat}[1]{\mathrm{Sat}_{#1}}

\newcommand{\liff}{\leftrightarrow}
\newcommand{\ntype}{\mathbf{N}}
\newcommand{\ltype}{\mathbf{L}}
\newcommand{\ttype}{\mathbf{T}}
\newcommand{\ID}{\mathsf{ID}}

\newcommand{\Lang}{\mathcal L}
\newcommand{\lang}[1]{\Lang_{#1}}
\newcommand{\langarith}{\lang 1}
\newcommand{\langsoarith}{\lang 2}
\newcommand{\langmuarith}{\lang \mu}
\newcommand{\langmuarithord}{\langmuarith^{\mathrm{ann}}}

\newcommand{\model}[1]{\mathfrak{#1}}

\newcommand{\interp}[2]{#2^{#1}}

\renewcommand{\models}{\vDash}

\newcommand{\setvarint}[1]{\wp{(\vert{#1}\vert )}}
\newcommand{\N}{\model {N}}

\newcommand{\seqcon}{ \mathcal{O}}

\newcommand{\infrule}{\mathsf r}
\newcommand{\supp}[1]{T_{#1}}

\newcommand{\id}{\mathsf{id}}
\newcommand{\wk}{\mathsf w}
\newcommand{\exch}{\mathsf{e}}

\newcommand{\cut}{\mathsf{cut}}

\renewcommand{\lor}{\vee}
\renewcommand{\land}{\wedge}
\newcommand{\limp}{\to}
\newcommand{\der}{\mathcal{D}}

\newcommand{\vldr}[2]{\vltr{#1}{#2}{\vlhy{\ \ \ }}{\vlhy{\ \ \ }}{\vlhy{\ \ \ }}}
\newcommand{\branch}{\mathfrak{b}}
\newcommand{\rs}[1]{RS(#1)}
\newcommand{\pr}{\mathsf{pr}}

\newcommand{\constr}{\mathcal{O}}

\newcommand{\dect}{\triangleleft}
\newcommand{\rp}{\mathsf{rp}}
\newcommand{\rpf}[1]{{#1}^\mathsf{rp}}

\newcommand{\ind}{\mathsf{ind}}
\newcommand{\xind}[1]{\ind(#1)}
\newcommand{\preaxiom}{\mathrm{pre}}
\newcommand{\indaxiom}{\mathrm{ind}}
\newcommand{\proves}{\vdash}

\newcommand{\IDL}[1]{\ID}

\newcommand{\CA}{\mathsf{CA}_0}

\newcommand{\ACA}{\mathsf A {\CA}}

\newcommand{\Lsys}{\mathsf L}
\newcommand{\Klass}{\mathsf K}
\newcommand{\LK}{\Lsys\Klass}
\newcommand{\LKeq}{\LK_=}

\newcommand{\PA}{\mathsf{PA}}

\newcommand{\muPA}{\mu\PA}

\newcommand{\Pin}[2]{\Pi^{#1}_{#2}}

\newcommand{\ca}[1]{#1\text{-}\CA}

\newcommand{\PSCA}{\ca{\Pin 1 2 }}

\newcommand{\WO}{\mathrm{WO}}

\newcommand{\circular}{\mathsf C}
\newcommand{\cmuPA}{\circular\muPA}
\newcommand{\muPAord}{\muPA_{\mathrm{ann}}}
\newcommand{\cmuPAord}{\cmuPA_{\mathrm{ann}}}

\title[Cyclic proof theory of generalised inductive definitions]{Cyclic proof theory of generalised inductive definitions}

\author[G.Curzi]{Gianluca Curzi\lmcsorcid{https://orcid.org/0000-0001-8746-1704}}
\address{University of Gothenburg}
\email{gianluca.curzi@gu.se}

\author[L. Melgaard]{Lukas Melgaard\lmcsorcid{https://orcid.org/0009-0008-0099-6932}}
\address{University of Birmingham}
\email{lksmlgrd@gmail.com}

\begin{document}

\begin{abstract}
We study cyclic proof systems for  $\muPA$, an extension of Peano arithmetic by generalised inductive definitions that is arithmetically equivalent to the (impredicative) subsystem of second-order arithmetic  $\PSCA$  by M\"{o}llerfeld.

The main result of this paper is that cyclic and inductive $\muPA$ have the same proof-theoretic strength.  First, we translate cyclic proofs into an annotated variant based on Sprenger and Dam's systems for first-order $\mu$-calculus, whose stronger validity condition allows for a simpler proof of soundness. We then  formalise this  argument  
within $\PSCA$, leveraging  M\"{o}llerfeld's  conservativity properties. To this end, we build on prior work by {Curzi} and Das on the reverse mathematics of the Knaster-Tarski theorem.

As a byproduct of our proof methods we show that, despite the stronger validity condition, annotated and ``plain'' cyclic proofs for $\muPA$ prove the same theorems. 

This work represents a further step in the non-wellfounded proof-theoretic analysis of  theories of arithmetic via impredicative fragments of second-order arithmetic, an approach initiated by Simpson's Cyclic Arithmetic, and continued by Das and {Melgaard} in the context of arithmetical inductive definitions. 
\end{abstract}

\maketitle

\section{Introduction}

Non-wellfounded proof theory explores generalized notions of proof that allow for infinite branches, where logical soundness is guaranteed by global specifications such as so-called  \emph{progressivity} (or \emph{trace}) condition. A central concept in this area is that of \emph{cyclic} (or \emph{circular}) proofs — non-wellfounded proofs that exhibit a regular tree structure often represented as finite   graphs. Compared to more traditional notions of proofs, cyclic proofs are particularly well-suited for formalizing (co)inductive reasoning, and they naturally subsume it. In particular, they provide a more ``analytic'' approach to (co)inductive invariants, which is crucial for the mechanisation of proof search~\cite{imp1,imp2,imp3,imp4}. 

The modern development of cyclic proofs can be traced back to the work of Niwiński and Walukiewicz \cite{NW96:games-for-mu}, who anticipated them in the context of tableau methods for the modal $\mu$-calculus. The topic was subsequently advanced by Simpson and Brotherston, who initiated the systematic investigation of cyclic proof theory~\cite{Brotherston05:cyc-prf,BrotherstonS11}.

Progressivity  condition is the most common validity criterion for cyclic proofs studied in the literature. It represents a quite robust and logic-independent requirement. However, because of its non-local (and non-modular) nature, checking progressivity  is in general  PSPACE-complete for cyclic proof systems (see, e.g.,~\cite{AfshariW24}).  More ``efficient'' alternatives to validity have been explored in~\cite{AfshariLeigh17:cut-free-mu,4,5,15}  with so-called \emph{(cyclic) reset proofs},  in which sequents can be annotated  to  keep track of ``progress''  in a local way. In particular, Sprenger and Dam introduced in~\cite{sprengerdam03:conf,sprengerdam03:journal} a reset system for first-order $\mu$-calculus where cyclic proofs are decorated with variables ranging over ordinals and equipped with  order-theoretic constraints on those ordinals. Their work was later refined by Leigh et al. in~\cite{Graham-cyclic-first-order-mu-calculus}.  The analysis of the trade-off between progressing and reset conditions has been recently thoroughly conducted in a general, abstract setting by Afshari and Wehr through the connection between Safra automata and reset proofs~\cite{DLeighW24}.

One of the primary goals of cyclic proof theory is to study the expressivity of cyclic proofs — for instance, with respect to provability, logical complexity, or proof-theoretic strength — and to compare them with standard, inductive(ly defined) proofs. Because of the non-wellfounded nature of cyclic proofs, this analysis  often requires ``indirect'' arguments, typically involving formalisation of mathematical principles in first-order or second-order arithmetical theories. In the context of first-order logic with inductive definitions, for example, Berardi and Tatsuta showed that inductive and cyclic proof systems do \emph{not} prove the same set of theorems~\cite{BerTat17:non-equivalence}, and that the equivalence between the two systems can be obtained when these are extended with Peano Arithmetic ($\mathsf{PA}$)~\cite{BerTat17:equivalence}.. The latter statement was  proven by  extracting an induction principle from the progressivity condition using (an extension of) Podelski-Rybalchenko termination theorem, and requires formalising  infinite Ramsey Theorem within $\mathsf{PA}$.

Berardi and Tatsuta's result was implicitly proven by Simpson, who independently showed that the cyclic variant of Peano Arithmetic ($\mathsf{CA}$) is equivalent to $\mathsf{PA}$~\cite{Simpson17:cyc-arith}. To establish this, he rather formalises the meta-theory of cyclic proofs in subsystems of second-order arithmetic using some results from reverse mathematics~\cite{Simpson99:monograph}.

Simpson's work was later refined by Das in two directions. First, he   showed in~\cite{Das20:ca-log-comp} that the logical complexity of inductive invariants  in $\mathsf{CA}$ is strictly simpler than $\mathsf{PA}$. Then, in joint work with {Melgaard}~\cite{Melgaard}, he studied the  extension of $\mathsf{PA}$ with arithmetic (finitely iterated) inductive definitions ($\mathsf{ID}_{<\omega}$) and showed that this system and its cyclic version $\mathsf{CID}_{<\omega}$ prove the same arithmetic statements, and so have the same proof-theoretic strength. 

    This paper continues the line of research initiated by Simpson, by exploring the cyclic proof theory of generalised inductive definitions. Its main focus  is $\muPA$, an  extension of $\mathsf{PA}$ by least and greatest fixed points of formulas (see~\cite{Curzi023}). Proof-theoretically, $\muPA$   was first studied by M\"{o}llerfeld~\cite{moellerfeld02:gen-ind-dfns}, inspired by Lubarsky’s work on ``$\mu$-definable sets''~\cite{Lubarsky1993definableSO}. The theory was later used by Curzi and Das to characterise the computational expressivity of  constructive fixed point logics~\cite{Curzi023}. Because of the presence of interleaving (or parametric) fixed points, the expressivity of $\muPA$ far exceeds that of $\mathsf{ID}_{<\omega}$. Indeed, M\"{o}llerfeld proved that $\muPA$ has the same proof-theoretic strength as  $\PSCA$, an (impredicative) subsystem of second-order arithmetic where the comprehension scheme is restricted to $\Pi^1_2$-formulas, one of the strongest theories for which we have an ordinal analysis~\cite{Rathjen95:recent-advances}. 

    The main result of this paper can be stated as follows:
    \begin{thm}\label{thm:main-result}
        Cyclic $\muPA$ ($\cmuPA$) is arithmetically equivalent to $\muPA$. 
    \end{thm}
     Our argument is summarized in~\autoref{fig:grand-tour-diagram}.  We first introduce a reset system for $\muPA$ inspired by~\cite{Graham-cyclic-first-order-mu-calculus}, called $\cmuPAord$. Compared to $\cmuPA$, $\cmuPAord$ is defined by a stronger  validity condition that allows for a more local and direct soundness proof. Despite being apparently more permissive, cyclic progressing proofs can be duly turned into cyclic reset proofs (\autoref{thm:translation}). We then adopt Simpson's method and we formalise the soundness argument for reset proofs within $\PSCA$,  using some results from reverse mathematics from~\cite{Curzi023} as well as reflection principles (\autoref{thm:formalisation-of-soundness} and \autoref{reflection}). 
     We then conclude by appealing to M\"{o}llerfeld's arithmetic conservativity of $\PSCA$ over $\muPA$~\cite{moellerfeld02:gen-ind-dfns}.  As a consequence of our proof methods, we also establish the equivalence between  progressing-based cyclic proofs and cyclic reset proofs, thus bridging the gap between two major routes to cyclic proof theory in the context of first-order $\mu$-calculi.

 \subsection*{Outline of the paper} This paper is structured as follows. In~\autoref{sec:2} we recall the syntax and the semantics of $\muPA$, as well as  its sequent calculus system. In~\autoref{sec 3} we define the cyclic proof system $\cmuPA$ and show that the latter can simulate $\muPA$ (\autoref{id-to-cid}). \autoref{sec 4}   introduces the  reset system $\cmuPAord$  and its soundness theorem (\autoref{thm:Graham-soundness}), while~\autoref{sec 5}  shows that cyclic progressing proofs of $\cmuPA$ can be translated into cyclic reset proofs (\autoref{thm:translation}). Finally, \autoref{sec 6} formalises the soundness argument for $\cmuPAord$ of~\autoref{sec 4} within $\PSCA$ (\autoref{thm:formalisation-of-soundness}).

    % \item Here we further extend to mupA. An extension is muPA, interleaving fixed points expressing positive inductive definitions (difference between ID<omega and muPA).
    % \item (first-order) mu-calculus introduced by  
    % Scott, Parker, Lubarsky, Mollerfeld (in particular Ratjen's proof theoretic ordinal),  Anupam and Curzi. 
    % \item We finally show that cmupa and mupa are the same. 
 % \item How we show that. Actually we use an intermediate step because of difficulty. Our advantage is to bridging the gap between two traditions in cyclic proof theory in the setting of arithmetic: progressivity vs annotated and reset systems with a simpler validity condition. Mention this tradition, Sprenger Dam, Graham's paper, Dominik etc (\cite{Graham-cyclic-first-order-mu-calculus,sprengerdam03:journal,  3,4,5,15}.). We use a variant of Graham's system.  Only few works attempt to relate annotated and non-annotated systems, and study the trade-off between the progressivity condition and the reset condition, as the latter produces a more simple soundness argument. Formalisation of soundness very easy, because no priority ordering and all that which are rather related with the finite presentation of cyclic proofs.
 % \item we bridge the gap between two traditions (similar to Dominik's but in a first order setting). Little explored in the literature. 
 %    \item structure of thepaper and contributions.
 %    \end{itemize}

\begin{figure}[t]
    \centering
\small\[
  \begin{tikzcd}[row sep=large]
      \muPA\vdash ``\phi \textrm{ true}"\arrow[d, swap,  "\mathrm{Thm.}~\ref{reflection}"]   & & \PSCA \vdash ``\phi \textrm{ true} " \arrow[ll, swap,  "\text{\cite{moellerfeld02:gen-ind-dfns}}" ]
      \\
         \muPA\vdash \phi  \arrow[r, swap,  "\mathrm{Thm.}~\ref{id-to-cid}"] &\cmuPA \vdash \phi   \arrow[r, swap, "\mathrm{Thm.}~\ref{thm:translation}"]& \cmuPAord \vdash \emptyset:\phi \arrow[u, swap, "\mathrm{Thm.}~\ref{thm:formalisation-of-soundness}"]
      \end{tikzcd} 
    % \begin{tikzcd}[row sep=large]
    %   \muPA\vDash ``\phi \textrm{ true}" \arrow[d, swap,  "(1)"]  & & \PSCA \vDash ``\phi \textrm{ true} " \arrow[ll, swap,  "\cite{mollerfeld}"] \\
    % \muPA\vDash \phi  \arrow[rr, swap,  "(2)"] & & \cmuPA \vDash \phi   \arrow[u, swap, "(3)"]
    %   \end{tikzcd} 
      \]
    \caption{Diagram illustrating the main results.}
    \label{fig:grand-tour-diagram}
\end{figure}

\subsection*{Preliminary definitions}
We assume that the reader is familiar with the syntax of the sequent calculus (e.g.~\cite{Basic-PT}) as well as with basic notions of graph theory. The following two definitions summarise the main definitions and notational conventions adopted in this paper.

\begin{defi}[Lists and trees]\label{defn:lists-and-trees}

A \emph{finite list over (a set) $X$} is a function $u:\{i \in \omega \mid i <n\}\to X$, often written $u_0 \cdots u_{n-1}$, where $n$ is called the \emph{length} of $u$. The set of lists over $X$ with length at most $n$ is denoted $X^{<n}$, and with $X^{<\omega}$ we denote the set of (finite) lists over $X$. With $\varepsilon$ we denote the list of length $0$. Given $u=u_0\cdots u_{n-1} \in X^{<\omega}$ and $a\in X$ we denote with $u::a$ the list $u_0\cdots u_{n-1}a$. 

A \emph{(possibly infinite, binary) tree} is a non-empty set $T \subseteq \Bool^{<\omega}$ that is prefix-closed, i.e., $b_1 \ldots b_n \in T$ implies $b_1\ldots b_i\in T$ for all $0 \leq i\leq n$. We call $\epsilon$ the \emph{root of $T$}, a maximal element w.r.t. the prefix order is called \emph{leaf},  and elements of the tree are called \emph{nodes}, which are denoted by $u,v,  \ldots$. A \emph{successor} of a node $u \in T$ is a node $ub \in T$ for some $b \in \{0,1\}$. A \emph{branch}  of $T$ is a subset $\branch \subseteq T$ such that $\epsilon\in \branch$ and, if $v\in \branch$ and $v$ has successors, then there is exactly one $b \in \{0,1\}$ such that $vb \in \branch$. Infinite branches are also represented as families $(v_i)_{i \in \omega}$ such that $v_0= \epsilon$ and $v_{i+1}= v_i b$ for some $b \in \{0,1\}$.
\end{defi}

\begin{defi}[Rules and (co)derivations]\label{def:preproofs}
    We consider \emph{(inference) rules} of the form $\vliiinf{\infrule}{}{\Gamma}{\Gamma_1}{\ldots}{\Gamma_n}$ with $n \in\{0,1,2\}$, where $\Gamma_1, \ldots, \Gamma_n$ are called \emph{premises} and $\Gamma$ is called \emph{conclusion}. A \emph{coderivation} over a set of rules $X$ is a pair $\der=(T, \ell)$ where $T$ is a tree called the \emph{support} of $\der$ and $\ell: T \to X$ is a label function such that, whenever $\ell(v)$ is a rule with $n$ premises $\Gamma_1,\dots, \Gamma_n$, then $v$ has precisely successors $v_1, \ldots, v_n$ and $\ell(v_i)$ is a rule with conclusion $\Gamma_i$ for all $1 \leq i \leq n$. A coderivation with finite support is called a \emph{derivation}. The \emph{conclusion} of a (co)derivation is the conclusion of the rule $\ell(\epsilon)$. Given some $u \in T_\der$, we write $\der_u$ for the (co)derivation with support $\supp{\der_u} =  \{v: uv \in T_\der \}$ given by $\ell(v) \dfn \ell(uv)$. We call such $\der_u$ a \emph{sub-(co)derivation} of $\der$. If $\der$ has only finitely many sub-coderivations, we call it \emph{cyclic} (or \emph{regular}). Cyclic coderivations $\der$ can be represented as a finite (possibly cyclic) graph.  
\end{defi}

\section{Arithmetical theories of generalised inductive definitions} \label{sec:2}

In what follows, we introduce the theory $\muPA$ and its semantics\footnote{This theory is often called $\mu$-calculus in the literature (see, e.g., \cite{de1969theory, moellerfeld02:gen-ind-dfns}. To date, however,  $\mu$-calculi represent a rather diverse family of systems. In this paper we adopt the formulation of the theory from~\cite{Curzi023}, which is referred to as  $\muPA$ to stress its arithmetic component.}. The theory can be constructed as a formal extension of  second-order arithmetic by defining  least fixed points $\mu X x \phi$  as relation symbols, as done in~\cite{moellerfeld02:gen-ind-dfns}. Instead we choose to treat them syntactically as ``binders'' and formulate the theory in a first-order fashion following~\cite{Curzi023}.

\subsection{Language of arithmetic with fixed points}

The \emph{language of $\muPA$}, written $\langmuarith$, is an extension of  the language of first-order arithmetic (with inequality)  $\langarith \dfn \{0,\succ, +,\times, <\}$ where, as usual, $0$ is a constant symbol (i.e.\ a $0$-ary function symbol), $\succ$ is a unary function symbol, $+$ and $\times$ are binary function symbols, and $<$ is a binary relation symbol. Specifically, the language of  $\muPA$ is obtained   from $\langarith$ by adding countably many set variables (i.e., predicate symbols)\footnote{W.l.o.g., in this paper we restrict the language  to second-order variables ranging over sets of individuals.}, written $X,Y$ etc, and fixed point binders $\mu$ and $\nu$.

\begin{defi}[Terms and (pre-)formulas]\label{defn:formulas}

The \emph{number terms}, the \emph{set terms}  and the \emph{pre-formulas} of $\muPA$ are simultaneously generated by the following grammars: 
\[
\begin{array}{rcll}
t, s&\dfn & x \ \vert \  0 \ \vert \ \succ \, t \ \vert \ s +  t \ \vert \ s \times t   & \text{(number terms)} \\[2ex]
\phi, \psi &\dfn &  s<t \ \vert \  s \nless t \ \vert \ s=t \ \vert \ s \neq t \ \vert \ t \in S \ \vert \ t \notin S \ \vert \qquad \\[1ex]   &&\phi \wedge \psi \ \vert \ \phi \vee \psi \ \vert \
\forall x \phi \ \vert \ \exists x\phi    & \text{(pre-formulas)} 
\\[2ex]
S &\dfn & X \  \vert \ \mu Xx\phi \ \vert \ \nu Xx \phi &  \text{(pre-set-terms)}
\end{array}
\] 
\todoinline{we need to be consistent about whether we write $\mu Xx \phi$ or $\mu Xx .\phi$}

The notions of \emph{free first-order variable} and \emph{free second-order variable} are standard. Notice that  $\nu$ and $\mu$ bind the free (individual and set) variables $X$ and $x$ in $\phi$. We denote with $FV_0 (\phi)$, $FV_0( t)$, and $FV_1( \phi)$
the set of free first-order variable in $\phi$ and $t$ and the set of free second-order variable in $\phi$, respectively.    We will use the standard notation for (capture-avoiding) first-order and second-order substitution  $\phi[S/X]$,  $\phi[s/x]$, and $t[s/x]$.

Pre-formulas of the form $s<t$, $s \nless t$, $s=t$, $s \neq t$, $t \in S$ and $t \notin S$ are called  \emph{atomic}.
While we don't have negation  in the language we use $\neg \phi$ as a shorthand in the following way:

\[
\arraycolsep=2pt
\def\arraystretch{1.3}
\begin{array}{rclrcl}
\neg s<t &:=& s \nless t  \qquad \qquad &\neg t \notin S &:=& t \in S  \\
 \neg s \nless t &:=& s < t & \neg s=t &:=& s \neq t  \\
\neg t \in S &:=& t \notin S  &\neg s \neq t &:=& s = t 
\end{array}
\]
\[
\arraycolsep=2pt
\def\arraystretch{1.3}
\begin{array}{rclrcl}
 \neg (\phi \wedge\psi ) &:=& \neg \phi \vee \neg \psi   \qquad \qquad& \neg \exists x \phi &:=& \forall x \neg\phi \\
 \neg (\phi \vee \psi ) &:=& \neg \phi \wedge \neg \psi &\quad \neg    t \in \mu X x  \phi( X)&:=& t \in \nu X x \neg \phi(\neg X) \\
 \neg \forall x \phi &:=& \exists x \neg\phi & \neg    t \in \nu X x  \phi( X) &:=& t \in \mu X x \neg \phi(\neg X)
\end{array}
\]

where $ y \in \neg X$ is shorthand for $\neg y \in X$.
Note that $\neg \neg \phi = \phi$.

Given a set variable $X$, we say that $X$ occurs \emph{positively} in a pre-formula $\phi$ if $\phi$ contains no pre-formulas of the form  $t \notin X$ for any term $t$. We say $\phi$ is a \emph{formula} if, for any sub-pre-formula in $\phi$ of the form $t \in \mu Xx\psi$ or $t \in \nu Xx\psi$, $X$ occurs positively in $\psi$. We say a pre-set-term $S$ is a set-term if $t \in S$ is a formula for some $t$.

\end{defi}

From now on, we will work with formulas where every second-order variable is bound. 
We adopt the standard abbreviations for connectives, such as $\phi \to \psi\dfn \neg \phi \vee \psi$, and $\phi \leftrightarrow \psi \dfn (\phi \to \psi) \wedge (\psi \to \phi)$, $\top= \phi \vee \neg \phi$ and $\bot\dfn \neg \top$.  We also write $\mathsf{n}$ for the \emph{numeral} of $n \in \mathbb{N}$, defined inductively by $\mathsf{0}\dfn 0$ and $\mathsf{n+1}\dfn \succ \mathsf{n}$. Finally, 
we will feel free of writing $\phi(\nu X x \phi(X), t) $ for $\phi[\nu X x \phi(X)/X, t/x] $.

\begin{exa}[Some fixed point formulas] \label{exmp:running-example-formulas}  Consider the following formulas of $\langmuarith$:
\begin{itemize} \setlength\itemsep{0.2cm}
    \item ${N}(X, x)\dfn x=0 \vee \exists y (y \in X \wedge x= \succ y)$
    \item ${L}(X, Y, x)\dfn x = \varepsilon \vee \exists z,w(z \in X \wedge w \in Y \wedge x= z:: w)$.
\end{itemize}
where the term $\varepsilon$ encodes the empty list, and the term $x::y$ encodes the operation of appending (the encoding of) an element $y$ to (the encoding of) a list $x$ (see~\autoref{defn:lists-and-trees})\footnote{With a small abuse of notation we will often not distinguish between lists and numbers coding lists in $\muPA$. Moreover, since the append operation is primitive recursive, we will treat $::$ as a binary function symbol, so     $x::y$ will be considered a term of $\langmuarith$.}. By definition, $\langmuarith$ contains the set terms $\ntype \dfn \mu Xx.  N(X,x)$, $\ltype( Y)\dfn  \mu Xx.  L(X, Y, x) $, and $\ttype\dfn \mu Yy. y \in \ltype (Y)$. The latter formula illustrates a key aspect of $\muPA$. Indeed, as we will see, thanks  to the interdependencies of fixed points it is possible to define mutually (co)inductive predicates,  a crucial source of expressivity for $\muPA$.
\end{exa}

\subsection{Interpreting generalised inductive definitions}

We now define the semantics of $\langmuarith$ and  the intended model. We will interpret  $\mu Xx \phi$ and $\nu Xx\phi$ as    least and greatest fixed points of the (parametrised) operator induced by $\phi$. To this end, we recall (a version of) the Knaster-Tarski theorem.

\begin{prop}
    [Knaster-Tarski on the power set lattice]
    \label{knaster-tarski-on-P(N)}
    Let $S$ be a set and let $f : \pow S \to \pow S$ be monotone, i.e.\ $A\subseteq B \subseteq S \implies f(A)\subseteq f(B)$.
    Then $f$ has a least fixed point $\mu f$ and a greatest fixed point $\nu f$. 
    Moreover, we have:
    \(
    \lfp f = \bigcap \{A\subseteq S \mid f(A) \subseteq A\}
    \)
    % \quad \text{and}\quad
    and
    \(
    \gfp f = \bigcup \{ A \subseteq S \mid A\subseteq f(A)\}
    \).
\end{prop}

Furthermore, least and greatest fixed points can be computed from their ordinal approximants:

\begin{prop}\label{prop:approximant}
     Let $S$ be a set and let $f : \pow S \to \pow S$ be monotone.  Then: 
\[
    \lfp f = \bigcup_{\kappa \in \mathrm{Ord}} \mu^\kappa f \quad \text{and} \qquad  \gfp f = \bigcap_{\kappa \in \mathrm{Ord}} \nu^{\kappa}f
    \]
  where the $\kappa$-approximant $\mu^\kappa f$ and $\nu^\kappa f$ are inductively defined by
  \[
 \mu^\kappa f= \bigcup_{\lambda <\kappa} f(\mu^\lambda f) \quad \text{and} \qquad \nu^\kappa f= \bigcap_{\lambda <\kappa} f(\nu^\lambda f) 
  \]
  In particular, $\mu^0f=\emptyset$ and $\nu^0f=S$.
\end{prop}

We shall henceforth adopt the notation of the proposition above, writing $\lfp f$ and $\gfp f$ for the least and greatest fixed point of an operator $f$, when they exist.

\begin{defi}[Standard model] \label{def:structure}
The \emph{standard model} for $\langmuarith$ is the $\langmuarith$-structure $\model{N}$, whose domain is $\Nat$, $ 0^{\model{N}}=0$, $\succ^{\model{N}}$ is the successor, $+^{\model{N}}$ is the addition, $\times^{\model{N}}$ is the multiplication, and $<^{\model{N}}$ is the natural (strict) order on natural numbers.
\end{defi}

\begin{defi}[Assignment] \label{def:assignment}
 An \emph{assignment} is a map $\rho$ associating an element $\rho{(x)} \in \vert \model N \vert $ with every  first-order variable $x$ and a set $\rho{(X)} \subseteq \model N$ with every  set variable $X$. Given an assignment $\rho$ and $a \in \vert \N \vert$ we write $\rho[x \mapsto a]$ for the assignment s.t. $\rho[x \mapsto a](x)\dfn  a$ and $\rho[x \mapsto a](y) \dfn  \rho(y)$ for all  $y \neq x$.  Similarly, for $A \subseteq \vert \N \vert$ we write $\rho[X \mapsto A]$ for the assignment that sends $X$ to $A$ and otherwise preserves the mapping given by $\rho$.
 \end{defi}

 %An \emph{assignment} for $\model{N}$ is a map $\rho$ associating an element $\rho{(x)} \in \vert \model M \vert $ with every free first-order variable $x$ and a set $\rho{(X)} \subseteq \model M$ with every  set variable $X$ . Given an assignment $\rho$ and $a \in \vert \N \vert$ we write $\rho[x \mapsto a]$ for the assignment s.t. $\rho[x \mapsto a](x)\dfn  a$ and $\rho[x \mapsto a](y) \dfn  \rho(y)$ for all  $y \neq x$.  Similarly, for $A \subseteq \vert \N \vert$ we write $\rho[X \mapsto A]$ for the assignment that sends $X$ to $A$ and otherwise preserves the mapping given by $\rho$.
 %\end{definition}

 \begin{defi}[Satisfaction] \label{def:satisfaction}
Given an assignment $\rho$ we extend $\rho$ to all terms as follows:
\begin{itemize}
\item $\rho(0)= 0$
    \item $\rho( f(t_1, \dots, t_k) ) \dfn \interp \N f (\rho (t_1), \dots, \rho (t_k))$ \quad $f \in \{\succ, +, \times\}$
\end{itemize}
We define the satisfaction relation $\model{N}, \rho \models \phi$ by induction on the formula $\phi$ as follows:
\begin{itemize} \setlength\itemsep{1em}
    \item $\model{N}, \rho \models t < u $ iff $\rho (t) < \rho  (u)$
    \item $\model{N}, \rho \models  t \nless u $ iff $\rho (t) \nless \rho  (u) $
     \item $\model{N}, \rho \models t = u $ iff $\rho (t) = \rho  (u)$
      \item $\model{N}, \rho \models  t \neq u $ iff $\rho (t) \neq \rho  (u)$
    \item $\model{N}, \rho \models t \in X$ iff $\rho (t) \in \rho{(X)}$
     \item $\model{N}, \rho \models  t \not \in X$ iff $\rho (t) \not \in \rho{(X)}$
     \item $\model{N}, \rho \models   \phi \wedge \psi$ iff $\model{N}, \rho \models    \phi$  and  $\model{N}, \rho  \models    \psi$ 
     \item $\model{N}, \rho \models   \phi \vee \psi$ iff $\model{N}, \rho  \models    \phi$  or  $\model{N}, \rho  \models    \psi$ 
     \item $\model{N}, \rho \models  \forall x\phi $ iff $\model{N}, \rho[x \mapsto a]  \models    \phi$  for all $a \in \vert \N \vert$
     \item $\model{N}, \rho \models  \exists x\phi $ iff $\model{N}, \rho[x \mapsto a]  \models    \phi$  for some $a \in \vert \model{N} \vert$
     \item $\model{N}, \rho \models  t \in \mu X x\phi $ iff $\rho (t) \in (\mu X x\phi)^\model{N}$  
      \item $\model{N}, \rho \models  t \in \nu X x\phi $ iff $\rho (t) \in (\nu X x\phi)^\model{N}$
\end{itemize}
where $(\mu X x\phi)^\model{N}=\mu   \formop \phi x \rho ^{\model{N}}$ and  $(\nu X x\phi)^\model{N}=\nu   \formop \phi x \rho ^{\model{N}}$ with $ \formop \phi x \rho ^{\model{N}} : A  \mapsto \{a \in \vert \N \vert \ \vert \ \N, \rho[x,X \mapsto a,A] \models \phi \}$. We simply write $\phi^{\model{N}}$ to denote $ \formop \phi x \rho ^{\model{N}}$ when it is clear from the context.
\end{defi}

The definition above is justified by the following property, using Knaster-Tarski theorem:

\begin{prop}\label{prop:monotonicity}
 Let $\rho$ an assignment, and $\phi(x, X)$ a formula positive in $X$. Then, $\phi^{\model{N}}: \setvarint \N \to \setvarint \N $ is a  monotone function. 
\end{prop}
\begin{proof}
    By induction on the structure of $\phi$. 
\end{proof}

\begin{cor}
Let $\phi(x, X)$ be a formula positive in $X$ and $\rho$ an assignment. Then:
 \[
 \def\arraystretch{1.5}
 \begin{array}{rcl}
      % $\model{N}, \rho \models  t \in \mu X x\phi $ iff $\rho t \in 
(\mu X x\phi)^\model{N}&=&\bigcap\{ A \subseteq \vert \N \vert  \ \vert \ \model{N}, \rho[X \mapsto A] \models \forall x ( \phi(x,X) \rightarrow x \in X  ) \}
\\ 
(\nu X x\phi)^\model{N}&=&\bigcup\{ A \subseteq \vert \N \vert  \ \vert \ \model{N}, \rho[X \mapsto A] \models \forall x (x \in X \rightarrow \phi(x,X))   \}
 \end{array}
\]
    
\end{cor}
\begin{proof}
Straightforward from~\autoref{prop:monotonicity}, Knaster-Tarski's theorem (cf.~\autoref{knaster-tarski-on-P(N)}) and \autoref{def:satisfaction}.
\end{proof}

\begin{prop}
 Let  $\phi$  be a  formula and $\rho$ as assignment. Then $\model{N}, \rho \models  \neg \phi$ iff $\model{N}, \rho \not \models    \phi$.
\end{prop}

\begin{exa}
Let $\rho$ be an assignment in the standard model $\N$. Revisiting~\autoref{exmp:running-example-formulas}, we have that:
\[
 \arraycolsep=1.7pt
\begin{array}{rclrclrcl}
  \mu^0N^{\model{N}}&=& \emptyset &\mu^0L^{\model{N}}&=& \emptyset&
     \mu^0T^{\model{N}}&=& \emptyset\\
     \mu^1 N^{\model{N}}&=& \{0\}&\mu^1L^{\model{N}}&=& \{\varepsilon\}&\mu^1T^{\model{N}}&=& \{\varepsilon\}\\
     \mu^2N^{\model{N}}&=& \{0,1\} & \mu^2L^{\model{N}}&=& \rho(Y)^{<1}&\mu^2T^{\model{N}}&=& (\mu^1T^{\model{N}})^{<\omega} \\
      \mu^3N^{\model{N}}&=& \{0,1,2\} & \mu^3L^{\model{N}}&=& \rho(Y)^{<2}&\mu^3T^{\model{N}}&=& (\mu^2T^{\model{N}})^{<\omega} \\
      &\vdots& & &\vdots& & &\vdots& \\
       \mu^\omega N^{\model{N}}&=& \mu^{\omega+1}N^{\model{N}}= \mathbb{N} \qquad & \mu^\omega L^{\model{N}}&=&\mu^{\omega+1}L^{\model{N}}=  \rho(Y)^{<\omega}  \qquad &\mu^\omega T^{\model{N}}&=& \mu^{\omega+1}T^{\model{N}}=\mathrm{HFL}
\end{array}
\]
 where  $\mathrm{HFL}$ is set of (the encodings of) hereditarily finite lists. Therefore, we have: $\ntype^{\model{N}}=(\mu Xx. N)^\model{N}=\mu   N^\model{N}= \mathbb{N}$, $\ltype(Y)^{\model{N}}= (\mu Xx. L(X, Y, x))^{\model{N}}=\mu  L^\model{N}=\rho(Y)^{<\omega}$, and $\ttype^{\model{N}}=(\mu Yy. T(Y, y))^{\model{N}}=\mu  T^{\model{N}}$. Alternatively, we can see $\ttype^{\model{N}}$ as the set of (the encodings) of rose trees, i.e., trees with unbounded (but finite) branching.

\end{exa}

\subsection{The theory $\muPA$}

The \emph{theory} $\muPA$   over the  language $\langmuarith$ as presented in this paper is essentially the first-order part of M\"ollerfeld's $\ACA(\langmuarith)$~\cite{moellerfeld02:gen-ind-dfns}, which extends  
Peano arithmetic (defined by the rules $\mathrm{PA}_1-\mathrm{PA}_8$ in~\autoref{fig:PA})  by the following axioms schemas for formulas $\psi(x)$ and $\phi(X,x)$ positive in $X$: 
    \begin{itemize} \setlength\itemsep{10pt}
        \item[] ($\preaxiom$-$\mu$) \quad $\forall y (\phi(\mu X x \phi, y) \limp y \in \mu X x \phi)$
        \item[] ($\indaxiom$-$\mu$)  \quad$\forall y (\phi(\psi, y) \limp \psi(y)) \limp \forall y (y \in \mu X x \phi  \limp \psi(y))$
         \item[] ($\preaxiom$-$\nu$) \quad$\forall y (y \in \nu X x \phi\limp \phi(\nu X x \phi, y))$
        \item[] ($\indaxiom$-$\nu$)  \quad$\forall y ( \psi(y)\limp \phi(\psi, y)) \limp \forall y (  \psi(y)\limp y \in \nu X x \phi )$\\
    \end{itemize}
    
The first (resp., third) axiom states that $\mu X x \phi$ (resp., $\nu X x \phi$) is a \emph{pre-fixed point} (resp., \emph{post-fixed point}). The second (resp., forth)  axiom states that it is the least among the  pre-fixed points (resp., the greatest of all post-fixed points). 

We will express $\muPA$ in sequent calculus style. In particular, thanks to the dualities of classical logic, we  adopt so-called ``one-sided'' presentation, thus omitting the turnstile symbol ``$\vdash$''.

\begin{defi}
    [Sequent calculus for $\muPA$]
   A \emph{sequent} $\Gamma $ is a list of formulas. The sequent  calculus system for  $\muPA$ extends $\LKeq$   (i.e., first-order logic with equality and substitution)  in~\autoref{fig:lkeq} by:
   
   \begin{itemize}\setlength\itemsep{0.2cm}
       \item the axioms $\mathsf{PA}_1$-$\mathsf{PA}_7$ and $\ntype$  in~\autoref{fig:PA}
       \item   the rules $\mu$ and $\xind{\phi}$  in~\autoref{fig:muPA}.
   \end{itemize} 
 We call $\wk$, $\exch$, $\theta\mathsf{-sub}$ $\cut$ \emph{structural rules}; $\id$, $\wedge$, $\vee$, $\forall$, $\exists$, $=$, $=\mathsf{-sub}$, $\nu$, $\mu$, $\mathsf{ind}(\phi)$ are called \emph{logical rules}; finally, \autoref{fig:PA} displays the \emph{arithmetic rules}.  In such rules,  $\Gamma$, $\Delta$, $\theta(\Gamma)$, $\Gamma(t)$, $\Gamma(s)$ are called \emph{contexts}, which are lists of formulas called \emph{side formulas}. {The formulas $\phi$ and $\psi$ in the rule $\exch$ are also called side formulas.} All the other formulas in the conclusion (resp., in a premise) of a rule  are called \emph{principal formulas} (resp.,  \emph{active}).
\end{defi}

It is easy to see that the sequent calculus system $\muPA$ derives the axioms ($\preaxiom$-$\mu$), ($\indaxiom$-$\mu$), ($\preaxiom$-$\nu$), ($\indaxiom$-$\nu$). For more details, see~\autoref{prop:axiomatic-implies-sequent} in~\autoref{app:derivations}. On the other hand, the following example shows derivability of the rule $\nu$ and that, in presence of $\xind{\phi}$,  numeric induction  and the rule $\ntype$ are equivalent.

\begin{rem}
    The one-sided sequent calculus streamlines the proofs. However, the more intuitive two-sided system will  still be used for examples to facilitate the reading, and so the notation $\Delta \vdash \phi $ will be considered as shorthand for $\neg \Delta, \phi$. As an example, the rule $\xind\phi$  becomes  \[\vliinf{}{\text{\footnotesize $ y$ fresh}}{\Delta , t \in \mu X  x \phi \vdash \Gamma  }
  {\Delta,  \psi(t) \vdash \Gamma }{\Delta, \phi(\psi, y) \vdash    \Gamma , \psi( y)   }\] Also, in the examples,  we will omit applications of structural rules as well as some arithmetical rules, and use double lines to abbreviate multiple applications of rules. 
\end{rem}

\begin{figure}
	\centering
	\[
	\begin{array}{c}
 \vlinf{\id}{
	}{\Gamma, \neg \phi,\phi}{}
	\qquad 
	\vlinf{\wk}{}{\Gamma, \Delta }{ \Gamma}
	\qquad
    \vlinf{\exch}{}{\Gamma,\phi, \psi, \Delta}{\Gamma,\psi, \phi ,\Delta}
    \qquad
	\vlinf{\theta \text{-} \mathsf{sub}}{}{ \theta(\Gamma) }{ \Gamma}
	\qquad
	\vliinf{\cut}{}{\Gamma}{\Gamma,  \phi}{\Gamma, \neg \phi}
	% \vlinf{\notl}{}{\Gamma, \lnot \chi \seqar \Delta}{\Gamma \seqar \Delta, \chi}
	% &
\end{array}
\]

 \bigskip
\[
\begin{array}{c}
\vliinf{\wedge}{}{\Gamma, \phi \land \psi}{\Gamma, \phi}{\Gamma, \psi}
\qquad
\vlinf{\vee}{}{\Gamma, \phi \lor \psi}{\Gamma, \phi,\psi}
\qquad 
\vlinf{\forall}{\text{\footnotesize $y$ fresh}}{\Gamma , \forall x \phi}{\Gamma, \phi[y/x]}
	\qquad 
	\vlinf{\exists}{}{\Gamma , \exists x \phi(x)}{\Gamma, \phi(t)}
    \end{array} 
	\] 
\bigskip
\[
\begin{array}{c}
 		\vlinf{= \text{-} \mathsf{sub}}{}{  \Gamma(t),\neg (s=t)}{\Gamma(s) }
	\qquad 
	\vlinf{=}{}{t=t}{}
\end{array} 
	\] 
	\caption{The sequent calculus $\LKeq$ for first-order logic with equality, where  $\theta$   is  a substitution, i.e.\ a map from variables to terms, extended to formulas and sequents in the expected way. } 
	\label{fig:lkeq}
\end{figure}

\begin{figure}\centering
  \[
  \vlinf{\mathsf{PA}_i}{}{\mathrm{PA}_i}{} \qquad  \qquad \qquad 
  \vlinf{\ntype}{}{t \in \ntype}{}
  \]
    \bigskip
    \[
    \def\arraystretch{1.2}
    \arraycolsep=1.4pt
\begin{array}{rcl}
\mathrm{PA}_1&\dfn &0 \neq \succ x\\
\mathrm{PA}_2&\dfn &	\succ x = \succ y \limp x=y\\
\mathrm{PA}_3&\dfn &x+0 =x\\
\mathrm{PA}_4&\dfn 	&x+ \succ y = \succ (x+y)
\end{array}\quad 
\begin{array}{rcl}
\mathrm{PA}_5&\dfn &x\times 0 = 0\\
\mathrm{PA}_6&\dfn 	&x\times \succ y = (x\times y) + x\\
\mathrm{PA}_7&\dfn &x<y \liff \exists z . x+\succ z=y\\
\mathrm{PA}_8&\dfn &(\phi(0) \land \forall x (\phi(x) \limp \phi(\succ x))) \limp \forall x \phi(x) \quad \ \phi \textrm{ arithmetical}
\end{array}
    \]
\caption{Sequent calculus axioms for $\PA$ and the rule $\ntype$ (see~\autoref{exmp:running-example-formulas}).}
	\label{fig:PA}
\end{figure}
 
\begin{figure}
    \centering
    \[
  \vlinf{ \nu}{}{\Gamma, t \in \nu X x \phi  }{\Gamma, \phi[\nu X  x \phi/X, t/x]} 
  \qquad 
  \vlinf{ \mu}{}{\Gamma, t \in \mu Xx \phi}{\Gamma,  \phi[\mu X  x \phi/X, t/x]}
		% \vlinf{\idl \phi}{}{\Gamma, \I \phi( t) \seqar \Delta }{\Gamma, \phi(\I \phi,  t) \seqar \Delta }
		% \qquad
		% \vlinf{\idr \phi}{}{\Gamma \seqar \Delta, \I \phi ( t)}{\Gamma \seqar \Delta , \phi(\I \phi,  t)}
	\]
 
  \bigskip
 \[
		\vliinf{\xind\phi}{\text{\footnotesize $ y$ fresh}}{\Gamma, t \in \nu X  x \phi  }
  {\Gamma,  \psi(t) }{\Gamma, \neg \psi( y) ,   \phi(\psi, y)   }
	\]
    \caption{Sequent calculus rules for least and greatest fixed points. } 
    \label{fig:muPA} 
\end{figure}

\begin{exa}[Post-fixed points and numerical induction]\label{exmp:derivability-examples}
The following derivations show that the rule $\nu$ from $\xind{\phi}$ and numerical induction can be subsumed by using $\xind{\phi}$ and the rule $\ntype$ in~\autoref{fig:PA}: 
  \[
\vlderivation{
\vliin{\xind{\phi}}{}{\Gamma, t \in \nu Xx\phi}{\vlhy{\Gamma, \phi(\nu Xx\phi, t)}}{
  \vliq{}{}{\Gamma, \neg \phi(\nu Xx\phi, y ), \phi(\phi(\nu Xx\phi, y ), y) }
    {
    \vlin{\mu}{}{\Gamma, y \in \mu Xx \neg \phi(\neg X), \phi(\nu Xx\phi, y)}
       {
       \vlid{=}{}{\neg \phi (\neg \mu Xx\neg \phi(\neg X), y), \phi(\nu Xx\phi, y)}
          {
          \vlin{\id}{}{\neg \phi (\nu Xx \phi, y), \phi(\nu Xx\phi, y)}{\vlhy{}}
          }
       }
    }
}
}
\]
\[
\vlderivation{
\vliq{\vee, }{}{t \not \in \ntype, ( \phi(0)\wedge \forall x(\phi(x)\to \phi(\succ x)))\to \forall x\phi(x)}
 {
 \vlin{\forall}{}{t \not \in \ntype, \neg \phi(0), \neg \forall x(\phi(x)\to \phi(\succ x)), \forall x\phi(x)}
  {
     \vlid{=}{}{t \not \in \ntype, \Gamma}
        {
         \vliin{\xind{\phi}}{}{\Gamma, t \in \nu Xx.\neg N(\neg X, x) }
     {\vlin{\id}{}{\Gamma, \neg \phi({y})}{\vlhy{}}}
           {
           \vliin{\wedge}{}{\Gamma,   \phi (z), \neg N(\neg\neg {\phi}, {z})}
          {
        \vlin{=\text{-}{\mathsf{sub}}}{}{\neg z=0, \Gamma, \phi(z)}
         {\vlin{\id}{}{\Gamma, \phi(0)}{\vlhy{}}}
           }
           {
         \vlin{=\text{-}{\mathsf{sub}}}{}{\neg z= \succ y, \neg \phi(y), \Gamma,  \phi(z)}
          {
          \vltr{}{\neg \phi(y), \Gamma, \phi({ \succ y})}{\vlhy{\ }}{\vlhy{\ }}             {\vlhy{\ }}
           }
           }
    }
  }
}
}
}
    \]
  In the latter derivation we let $\Gamma = \forall x (\phi(x) \to \phi   (\succ x) )$. Notice that the double lines labelled with $\phi$ express a ``functoriality'' property (see~\autoref{prop:funct-ind} in~\autoref{app:derivations}).

\end{exa}

In fact, the choice of $\ntype$ over numerical induction has the only purpose of simplifying the definition of progressing trace in the next section.

\section{Cyclic proofs for the theory of generalised inductive definitions}\label{sec 3}

In this section we introduce a cyclic version of the theory $\muPA$, based on ideas from the modal $\mu$-calculus \cite{NW96:games-for-mu,Studer08:pt-for-mu,AfshariLeigh17:cut-free-mu} and calculi of first-order logic with inductive definitions \cite{Brotherston05:cyc-prf,BroSim06:seq-calc-inf-desc}. Notice that this section borrows notions and notation from~\autoref{def:preproofs}.

\subsection{Progressivity and cyclic proofs} 

 We can define a notion of non-wellfounded proof (i.e., infinite but finitely branching proof) which we call coderivation:

\begin{defi}[Coderivations]
A \emph{($\muPA$-)coderivation}, is a coderivation over the set of rules extending  the rules of $\LKeq$ (\autoref{fig:lkeq}) by: 
\begin{itemize}\setlength\itemsep{0.2cm}
    \item the arithmetic rules $\mathsf{PA}_1$-$\mathsf{PA}_7$ and $\ntype$ in \autoref{fig:PA}.
    \item the fixed point rules $\mu$ and $\nu$ in~\autoref{fig:muPA}.
\end{itemize}
\end{defi}

Intuitively, ($\muPA$-)coderivations are built from the same set of rules as $\muPA$ except that the rule $\nu$ replaces fixed point induction $\xind\phi$. We will show that the latter can be subsumed by the former in the context of cyclic proofs. 

We now adapt to our setting a well-known validity  criterion from non-wellfounded proof theory.

\begin{defi}
[Traces and progress]
Let $\der$ be a coderivation, and  let $(v_i)_{i \in \omega}$ be an infinite branch along $ \der$. A \emph{trace} along $(v_i)_{i\in \omega}$ is a sequence of {occurrences} formulas ${\tau=}(\phi_i)_{i\geq k}$ such that for all $i\geq k$, $\phi_i$ is in the conclusion of a rule $\infrule$, $\phi_{i+1}$ is in a premise of $\infrule$, and one of the following cases applies:
\begin{itemize}
\item  $\phi_i$ is the principal formula and $\infrule=\exch$  with conclusion $\Gamma, \phi_i, \psi, \Delta$ and  premise  $\Gamma, \psi, \phi_{i+1}, \Delta$ 
\item $\phi_i$ is the principal formula and $\infrule=\exch$  with conclusion $\Gamma, \psi, \phi_i, \Delta$ and premise  $\Gamma, \phi_{i+1}, \psi, \Delta$
\item  $\phi_i$ is the principal formula with $\infrule\neq\exch$ and   $\phi_{i+1}$ is an active formula
\item  $\phi_i$ and  $\phi_{i+1}$ are the $n$-th side formulas in the context, for some $n$\footnote{Notice that in this case  $\phi_i$ might be obtained from $\phi_{i+1}$ by performing substitution of terms.}.
\end{itemize}
We say that $\phi_{k+1}$ is an \emph{immediate ancestor} of $\phi_k$ if they extend to some trace $(\phi_i)_{i\geq k}$.
If $\phi_i = t \in \nu Xx.\phi$ (respectively $\phi_i = t \in \mu Xx.\phi$) for some $t$, the rule $\infrule$ is $\nu$ (respectively $\mu$) and $\phi_i$ is principal then we say $\nu Xx.\phi$ (respectively $\mu Xx.\phi$) progresses at $i$. We say a trace $(\phi_i)_{i\geq k}$ is \emph{progressing} if there is a $\nu$-term $\nu Xx.\phi$ such that $\nu Xx.\phi$ progresses infinitely often and $\nu Xx.\phi$ is minimal (with respect to number of symbols) among all set-terms that progress infinitely often.

\end{defi}

 % $\phi_i$ is the $n$-th side formula,   $\infrule$ is $\theta\mathsf{-sub}$ and $\theta(\phi_{i+1})=\phi_i$
 %     or $\infrule$ is  $\text{-}\mathsf{sub}$  with principal formula $s=t$  and, for some $\psi(x,y)$, we have $\phi_{i+1} = \psi(s)$ and $\phi_i = \psi(t)$
 %     or  $\phi_{i+1}=\phi_i$ and  $\phi_{i}$  (resp., $\phi_{i+1}$) is the $n$-th formula of the context $\Gamma$ in the conclusion (resp., premise) of $\infrule$.  

The following is a well-known property about traces:

\begin{prop}[Folklore]\label{prop:equivalent-trace-condition}
    Let $\der$ be a coderivation. If  $(\phi_i)_{i \in \omega}$ is a progressing trace and $\nu Xx\psi$ is an infinitely often progressing set-term with minimal size then:
    \begin{itemize}
        \item for infinitely many $i$, $\phi_i=t \in \nu Xx\psi$ (for some $t$),
        \item for all but finitely many $i$, $\phi_i$ contains $\nu Xx \psi$ as a subterm.
    \end{itemize}
    % \gianluca{Is $\nu$ term is required to be closed? If so we add somewhere that we only consider closed fixed points.}
\end{prop}
\todoinline{Check set term substitution}

\begin{defi}
[Progressing and cyclic proofs]
A \emph{progressing proof} (also called \emph{trace proof}) is a coderivation $\der$ for which each infinite branch has a progressing trace.   
  We write $\cmuPA$ for the class of cyclic progressing proofs and $\cmuPA \vDash \phi$  if there is a cyclic progressing proof of $\phi$ (see~\autoref{def:preproofs}).
\end{defi}

Henceforth, we may mark two rule steps by a special symbol (e.g., $\bullet$ or $\star$)  in a cyclic coderivation to indicate roots of identical sub-coderivations.

\begin{exa} The cyclic coderivation   \[\vlderivation{
\vlin{\mu}{\bullet}{t \in \mu Xx (x\in X)}{\vlin{\mu}{\bullet}{t \in \mu Xx(x\in X)}{\vlhy{\vdots}}}
}
\] is not progressing, as its (only) branch looping at $\bullet$ has no $\nu$-formula, and so any trace traversing such branch  cannot be progressing. On the other hand the following is an example of a cyclic progressing proof (see~\autoref{exmp:running-example-formulas}):

\[
\vlderivation{
\vlin{\nu}{\bullet}{\orange{ x  \in \ttype}\vdash x \in \ntype}
  {
  \vlin{\nu}{\star}{\orange{ x  \in \ltype( \ttype)} \vdash x \in \ntype}
    {
    \vliiq{\wedge, \vee, \forall}{}{\orange{ L( \ltype( \ttype),  \ttype, x)}\vdash x \in \ntype}
     {
     \vlin{=\text{-}\mathsf{sub}}{}{\blue{z  \in \ltype( \ttype)} , \red{w    \in \ttype}, x= z ::w\vdash  x \in \ntype}
       {
        \vliq{\text{arithmetic}}{}{\blue{z  \in \ltype( \ttype)}, \red{w    \in \ttype}\vdash   z ::w  \in \ntype}
           {
           \vliin{\wedge}{}{\blue{z  \in \ltype( \ttype)}, \red{w    \in \ttype}\vdash  z \in \ntype \wedge w \in \ntype}
             {
             \vlin{\theta\text{-}\mathsf{sub}}{}{\blue{z \in \ltype(\ttype)}\vdash z \in \ntype}{
          \vlin{}{\star}{\blue{x \in \ltype(\ttype)}\vdash x \in \ntype}{\vlhy{\vdots}}
          }
             }{
             \vlin{\theta\text{-}\mathsf{sub}}{}{\red{w  \in \ttype}\vdash  w \in \ntype}
             {
             \vlin{}{\bullet}{\red{x  \in \ttype}\vdash  x \in \ntype}{\vlhy{\vdots}}
             }
             }
           % {
           % \vltr{\der'}{\neg (z\in \ntype \wedge w \in \ntype),  \neg  z ::w  \in \ntype}{\vlhy{}}{\vlhy{}}{\vlhy{}}
           % }
        }
       }
     }
         {
    \vltr{\der}{ x=\varepsilon \vdash  x \in \ntype}{\vlhy{}}{\vlhy{}}{\vlhy{}}}
    }
  }
}
\]

The coderivation syntactically shows that $\ttype^\N\subseteq \ntype^\N$, i.e., 
that (the encodings of) hereditarily finite lists are   natural numbers. Despite the statement being tautological in nature, the formal argument is non-trivial and requires nesting of cycles, due to interleaving of fixed points defining $\ttype$. 

Now, for any infinite branch $\branch$ one of the following cases applies:
\begin{itemize}
    \item $\branch$ eventually loops only at $\star$, in which case the trace given by the alternation of $\orange{orange}$ and  $\blue{blue}$ formulas is progressing and the smallest $\nu$-term unfolding infinitely often is $\ltype(\ttype)$ (in negative position). 
    \item $\branch$ eventually loops only at $\bullet$, in which case the trace given by the alternation of $\orange{orange}$ and  $\red{red}$ formulas is progressing and the smallest $\nu$-term unfolding infinitely often is $\ttype$  (in negative position). 
    \item $\branch$ loops infinitely often at both $\star$ and $\bullet$, in which case the alternation of the $\orange{orange}$ formulas,  the  $\blue{blue}$ formulas (looping at $\star$), and the $\red{red}$ formulas (looping at $\bullet$) is progressing and the smallest $\nu$-term unfolding infinitely often is   $\ttype$ (in negative position).
    % \begin{itemize}
    %     \item at every loop at $\star$ the trace follows the \blue{blue} formulas
    %     \item at every loop at $\bullet$ the trace follows the \red{red} formulas
    %     \item if a loop at $\bullet$ is followed by a loop at $\star$ then the trace following the \red{red} formulas  switches to the  \blue{blue} formulas
    %     \item if a loop at $\star$ is followed by a loop at $\bullet$ then the trace following the \blue{blue} formulas  switches to the \orange{orange} formulas and from those to the \red{red} ones.
    % \end{itemize}
\end{itemize}
Therefore, any infinite branch has a progressing trace, and so the coderivation  above is a cyclic progressing proof.
\end{exa}

\subsection{Simulating inductive proofs}
Our cyclic system $\cmuPA$ subsumes $\muPA$ by a standard argument:
\begin{thm}
[Induction to cycles]
\label{id-to-cid}
If $\muPA \proves \phi$ then $\cmuPA \proves \phi$.
\end{thm}
\begin{proof}
[Proof sketch]
    We proceed by induction on the structure of a $\muPA$ proof.
    The critical step is $\xind \phi$, for which we do not have a corresponding rule in $\cmuPA$.
    We simulate this rule by the following cyclic progressing proof:
  \[
    \vlderivation{
    \vliin{\cut}{}{\Gamma,  t \in \nu X x \phi }
    {
        \vlhy{\Gamma, \psi(t)}
    }
    {
        \vlin{\nu}{\bullet}{\Gamma,\neg \psi(t), \blue{t \in \nu X x \phi}}{
        \vliin{\cut}{}{\Gamma,\neg\psi(t), \blue{\phi(\nu Xx \phi, t)}}{
            \vliq{\phi}{}{\Gamma, \blue{\phi(\nu Xx \phi, t)},\neg \phi(\neg \psi,t)}{
            \vlin{\nu }{\bullet}{\Gamma, \blue{t \in \nu X x \phi } , \neg \psi(t)}{\vlhy{\vdots}}
            }
        }{
            \vlin{[t/y]}{}{\Gamma, \neg \psi(t), \phi(\psi,t)}{\vlhy{\Gamma, \neg\psi(y), \phi(\psi,y)}}
        }
        }
    }
    }
    \]
  where  the rule steps  marked $\phi$ can be defined by induction on the structure of $\phi$ (see~\autoref{prop:funct-cycl} in~\autoref{app:cmupa}).
 % \autoref{functoriality-cyclic-iid}, under substitution of $\I\phi$ and $\psi$ for $Y$ and $Z$ respectively.
    To check progressivity of the above coderivation, notice that any infinite branch is either progressing by the induction hypothesis, or loops infinitely on $\bullet$ and has the progressing trace coloured in \blue{blue}.
\end{proof}

\section{A reset proof system for $\muPA$}\label{sec 4}

In this section we introduce an alternative route to cyclic proofs that allows for a more local and simpler (although more constrained) 
validity condition: so-called \emph{reset proofs} (e.g.,~\cite{AfshariLeigh17:cut-free-mu,4,5,15}. Specifically, we define a \emph{(cyclic) reset proof system} for $\muPA$, called $\cmuPAord$,  inspired by the annotated systems for first-order $\mu$-calculus in~\cite{sprengerdam03:journal,Graham-cyclic-first-order-mu-calculus}. A key property of reset proof systems is the introduction of variables ranging over ordinals annotating fixed point formulas in a coderivation. Intuitively, an annotated fixed point   $ \mu^\kappa Xx\phi$, where $\kappa$ is an ordinal variable,   expresses syntactically an ordinal approximant of  $\mu Xx\phi$ (see~\autoref{prop:approximant}).  The advantage of this approach is the possibility to internalise certain features of the progressivity condition, which can be then subsumed in a much simpler condition called \emph{reset condition}. This allows for a more straightforward soundness argument. Soundness for $\cmuPA$ can be then established by defining a translation of progressing proofs of $\cmuPA$  into reset proofs.

\subsection{Syntax and semantics of $\muPAord$}

The \emph{language of $\muPAord$}, written $\langmuarithord$, is a three-sorted language obtained by augmenting  $\langmuarith$  by a {countable set of \emph{ordinal variables}  $\mathcal{OV}$} ranging over $\kappa, \lambda$. We extend~\autoref{defn:formulas} to include new set terms $\mu^\kappa Xx\phi$ and $\nu^\kappa Xx\phi$ for any ordinal variable $\kappa$.

Concerning semantics, the standard model is defined as in~\autoref{def:structure}, while assignments $\rho$ in~\autoref{def:assignment} are extended to associate with every ordinal variable $\kappa \in \mathcal{OV}$ an ordinal $\rho(\kappa)$. The notion of satisfaction in~\autoref{def:satisfaction} is then extended to the clauses $ t \in \mu^\kappa X x\phi$ and  $ t \in \nu^\kappa X x\phi$ as follows:

\begin{itemize}\setlength\itemsep{0.2cm}
 \item  $\model{N}, \rho \models  t \in \mu^\kappa X x\phi $ iff $\rho (t) \in \mu^{\rho(\kappa)}   \formop \phi x \rho $,
    \item  $\model{N}, \rho \models  t \in \nu^\kappa X x\phi $ iff $\rho( t) \in \nu^{\rho(\kappa)}   \formop \phi x \rho $.
\end{itemize}

All the remaining definitions and notational conventions smoothly adapt to the new setting.

\subsection{Constraints}

 Our reset proof system  derives special sequents,  endowed with  requirements on ordinal variables.

 \begin{defi}[Constraints and extended sequents]\label{defn:constraints-extended-sequents}
     A \emph{constraint} $\seqcon$ is a pair $(V_{\seqcon}, <_{\seqcon})$ where $V_{\seqcon}$ is a finite set of ordinal variables and $<_{\seqcon}$ is a strict partial order such that for each $\kappa \in V_{\seqcon}$, the set $\{ \lambda \in V_{\seqcon} \ | \ \kappa <_{\seqcon} \lambda \} $ is linearly ordered by
$<_{\seqcon}$. We sometimes write $\kappa \in \seqcon$ as shorthand for $ \kappa \in V_{\seqcon}$ and we write $\seqcon \proves \lambda < \kappa$ instead of $\lambda <_{\seqcon} \kappa$. We also write $\lambda \leq_{\seqcon} \kappa$ for $\lambda <_{\seqcon} \kappa$ or $\lambda = \kappa$. With a small abuse of notation, we denote with $\emptyset$ 
the constraint $\seqcon$ such that $V_{\seqcon}= \emptyset$.

An \emph{(annotated) sequent} $\Gamma$ is a finite {list} of formulas of $\muPAord$.  An \emph{extended sequent} is an expression $ \seqcon : \Gamma$ where $\Gamma$ is an annotated sequent and $\seqcon$ is a constraint such that $V_{\seqcon}$ contains all ordinal variables in $\Gamma$. Given an annotated sequent $\Gamma$, we denote with $\Gamma^-$ the sequent obtained by stripping away the ordinal variables appearing in $\Gamma$.
 \end{defi}

Intuitively, the condition $\lambda <_\constr \kappa$ is interpreted as the assumption that $\lambda$ is smaller than $\kappa$. As a consequence,  $\seqcon$ can be seen as a collection of trees with $V_{\seqcon}$ the set of vertices and the descendant relation the converse of the order $<_{\seqcon}$: the children of $\kappa \in V_{\seqcon}$ are then $\lambda \in V_{\seqcon}$ where $\lambda <_{\seqcon} \kappa$ and for no $\xi \in V_{\seqcon}$ do we have $\lambda <_{\seqcon} \xi <_{\seqcon} \kappa$.

The semantics for $\langmuarithord$ can be generalised to extended sequents in the obvious way. Given a constraint $\constr$, and assignment $\rho$, we define $\rho \models \constr$ iff $\lambda <_\constr \kappa$ implies $\rho(\lambda)<_\constr \rho(\kappa)$ for all $\lambda, \kappa \in \constr$. An extended sequent $\constr: \Gamma$ is \emph{satisfied in $\model{N}$} if, for all  assignments $\rho$ for $\model N$, $\rho \models \constr$ implies $\model{N}, \rho \models \bigvee \Gamma$.

\subsection{Annotated coderivations and reset proofs}

For the purpose of presenting the inference rules we require three operations on constraints:
$$\seqcon + \lambda := (V_{\seqcon} \cup \{ \lambda \} , <_{\seqcon})$$ 
$$ \seqcon +_{\kappa} \lambda := (V_{\seqcon} \cup  \{ \lambda \} , <_{\seqcon} \cup \ \{ (\lambda, \xi) \ | \ \kappa \leq_\constr \xi \} ) $$
$$\seqcon \setminus V := (U, <_{\seqcon} \cap \ U^2 ) \text{ where }V \subset V_{\seqcon} \text{ and } U = V_{\seqcon} \setminus V$$
The first operation corresponds to adding $\lambda$ as a fresh ordinal variable which becomes the root of a new tree in $\seqcon$. In the second operation $\lambda$ is added as a fresh child of the ordinal variable $\kappa$. Note that both of these operations preserve the condition that for each $\kappa \in V_{\seqcon}$, the set $\{ \lambda \in V_{\seqcon} \ | \ \kappa <_{\seqcon} \lambda \} $ is linearly ordered by
$<_{\seqcon}$ and so tree structure is also preserved.  The final operation is that of
restricting a constraint set to a specified subset of ordinal variables.  We write $\seqcon_1 \subseteq \seqcon_2$ to express that $\seqcon_1 = \seqcon_2 \setminus V$ for some set $V \subseteq V_{\seqcon_2}$.

\begin{figure}
	\centering
	\[
	\begin{array}{c}
 \vlinf{\id}{
	}{\constr:\neg \phi,\phi}{}
	\qquad 
	\vlinf{\wk}{}{\constr:\Gamma, \Delta }{\constr: \Gamma}
	\qquad
    \vlinf{\exch}{}{\constr:\Gamma,\phi, \psi, \Delta}{\constr:\Gamma,\psi, \phi ,\Delta}
    \qquad
	\vlinf{\theta \text{-} \mathsf{sub}}{}{ \constr:\theta(\Gamma) }{ \constr:\Gamma}
	\qquad
	\vliinf{\cut}{}{\constr:\Gamma}{\constr:\Gamma,  \phi}{\constr:\Gamma, \neg \phi}
	% \vlinf{\notl}{}{\Gamma, \lnot \chi \seqar \Delta}{\Gamma \seqar \Delta, \chi}
	% &
\end{array}
\]

 \bigskip
\[
\begin{array}{c}
\vliinf{\wedge}{}{\constr:\Gamma, \phi \land \psi}{\constr:\Gamma, \phi}{\constr:\Gamma, \psi}
\qquad
\vlinf{\vee}{}{\constr:\Gamma, \phi \lor \psi}{\constr:\Gamma, \phi,\psi}
\qquad 
\vlinf{\forall}{\text{\footnotesize $y$ fresh}}{\constr:\Gamma , \forall x \phi}{\constr:\Gamma, \phi[y/x]}
	\qquad 
	\vlinf{\exists}{}{\constr:\Gamma , \exists x \phi(x)}{\constr:\Gamma, \phi(t)}
    \end{array}
\] 
\bigskip
\[
 \begin{array}{c} 
   		\vlinf{= \text{-} \mathsf{sub}}{}{  \constr:\Gamma(t),\neg s=t}{\constr:\Gamma(s) }
	\qquad 
	\vlinf{=}{}{\constr:t=t}{}
\qquad 
\overline{ F(\overline{C}) }  \vlinf{\mathsf{PA}_i}{}{\constr:\mathrm{PA}_i}{1 \leq i \leq 7}
  \qquad 
  \vlinf{\ntype}{}{\constr:t \in \ntype}{}
     \end{array}
\]
\bigskip
  \[
	\begin{array}{c}  
    \vlinf{ \mu}{}{ \seqcon : \Gamma, t \in \mu Xx \phi}{\seqcon : \Gamma,  \phi[\mu X  x \phi/X, t/x]}
    \qquad
    \vlinf{ \nu}{\text{\footnotesize $\lambda$ fresh}}{\seqcon : \Gamma, t \in \nu X x \phi  }{\seqcon + \lambda:  \Gamma, t \in \nu^{\lambda} X x \phi } 
\end{array}
\]
\bigskip
  \[
	\begin{array}{c}  
 \vlinf{ \nu^{\kappa}}{\text{\footnotesize $\lambda$ fresh}}{\seqcon : \Gamma, t \in \nu^{\kappa} X x \phi  }{\seqcon +_{\kappa} \lambda : \Gamma, \phi[\nu^{\lambda} X  x \phi/X, t/x]}  \qquad \vlinf{\mathsf{cw}}{}{ \seqcon :  \Gamma }{ \seqcon \setminus L :  \Gamma}
 \end{array}
\] 
\caption{The inference rules defining annotated coderivations  of $\muPAord$, where $\mathsf{PA}_i$ are the axioms of Peano arithmetic in~\autoref{fig:PA}.} 
	\label{fig:mupaord rules-extended}
\end{figure}

\begin{defi}[Annotated coderivations]
An  \emph{annotated coderivation} is a coderivation over  the set of  rules in~\autoref{fig:mupaord rules-extended}.
Notice that in the rules $\nu$ and $\nu^\kappa$ we require  that the ordinal variable $\lambda $ does not occur in the conclusion.    
\end{defi}

We will utilise as additional rules two  special cases of the constraint weakening rule ($\mathsf{c}\wk$). The first rule is called \emph{reset rule} (\autoref{fig:reset-pruning}, left) and requires that  $\kappa$ does not occur in $\Gamma$. The set $C(\kappa)$ consists of all and only the children of $\kappa$ in $\constr$ and is non-empty. Note that $C(\kappa)$ will typically \emph{not} consist of all descendants of $\kappa$; only the immediate children of $\kappa$ are removed by an application of the reset rule. Notice also that the definition of extended sequent ensures that no ordinal variable in $C(\kappa)$ occur in any formula in $\Gamma$. 

The second rule is a novelty of this paper, and  is  called  \emph{pruning rule} (\autoref{fig:reset-pruning}, right). The set  $P(\Gamma)$  should be non-empty and consists of all ordinal variables such that all their descendants do not occur in $\Gamma$ (i.e., $\{ \kappa \in \constr \ | \ \forall \lambda \in \constr (\lambda \leq_{\constr} \kappa \implies \lambda \text{ does not occur in $\Gamma$}) \}$). In other words, the pruning rule discards all the ordinal variables that are ``dead ends'', in the sense that none of their descendants play any role in the annotated coderivation.

While the reset rule is the key ingredient for defining the validity condition on annotated coderivations, we introduce the pruning rule specifically to avoid unbounded increase of ordinal variables in infinite branches of annotated coderivations.

We are now ready to define the reset condition, which represents a counterpart to the progressivity condition in the setting of annotated coderivations.

\begin{defi}[Reset proofs] Let $\der$ be an annotated coderivation, and let $\branch=(\constr_i: \Gamma_i)_{i \geq 0}$ be an infinite branch of $\der$. We say that $\kappa$ is a \emph{stable variable} in $\branch$ if there is $i_0 \geq 0$ such that  $\kappa \in \constr_{i}$, for every $i \geq i_0$. We say that $\kappa$ is a \emph{reset variable} in $\branch$ if there are infinitely many occurrences of $\rs \kappa$ in $\branch$. 

We say that $\der$ is a  \emph{reset proof} if every infinite branch has a stable reset variable. We denote with  $\cmuPAord$ the class of  cyclic reset proofs. 
\end{defi}

\begin{exa}
Consider the following annotated coderivation of the formula $t \in \nu Xx (x \in X)$.
$$
\vlderivation{
\vlin{\nu}{}{\emptyset: t \in \nu Xx(x \in X)}
  {
  \vlin{\nu^{\lambda}}{}{(\{\lambda\}, \emptyset): t \in \nu^{\lambda} Xx(x \in X)}
  {
  \vlin{\nu^{\kappa}}{\bullet}{(\{\lambda, \kappa\}, \{(\kappa, \lambda)\}): t \in \nu^{\kappa} Xx(x \in X)}
  {
   \vlin{RS(\lambda)}{}{(\{\lambda, \kappa, \theta\}, \{(\kappa, \lambda), (\theta, \kappa), (\theta, \lambda)\}): t \in \nu^{\theta} Xx(x \in X)}
  {
  \vlin{\nu^{\theta}}{}{(\{\lambda,  \theta\}, \{(\theta, \lambda)\}): t \in \nu^\theta Xx(x \in X)}
  {
  \vlin{RS(\lambda)}{}{(\{\lambda, \theta, \kappa\}, \{(\theta, \lambda), (\kappa, \theta), (\kappa, \lambda)\}): t \in \nu^{\kappa} Xx(x \in X)}
  {
  \vlin{\nu^\kappa}{\bullet}{(\{\lambda, \kappa\}, \{(\kappa, \lambda)\}): t \in \nu^{\kappa} Xx(x \in X)}{\vlhy{\vdots}}
  }
  }
  }
  }
  }
  }
  }
$$ 
It has only one infinite branch (looping at $\bullet$) that unfolds infinitely often (annotations of) the formula $t \in \nu Xx (x \in X)$. The branch contains a stable reset variable $\lambda$. Therefore, the annotated coderivation is a (cyclic) reset proof.
\end{exa}

\begin{figure}
    \centering
    \[
\vlinf{\rs \kappa}{}{\constr: \Gamma}{\constr \setminus C(\kappa): \Gamma}
\qquad
\vlinf{\pr}{}{\constr: \Gamma}{\constr \setminus P(\Gamma): \Gamma} 
    \]
    \caption{From left, the reset rule and the pruning rule.}
    \label{fig:reset-pruning}
\end{figure}

\begin{rem}
As previously mentioned,  $\muPAord$ is based on the reset proof system for the $\mu$-calculus developed in~\cite{sprengerdam03:journal}, and later refined in~\cite{Graham-cyclic-first-order-mu-calculus}. Compared to these systems, $\muPAord$  differs in many respects. First, concerning  syntax, $\muPAord$ is endowed with the language of Peano arithmetic and omits any quantification over ordinal variables. Concerning the proof system, $\muPAord$ is formulated in one-sided sequent calculus style. Moreover, we  avoid using finite representations of cyclic proofs in terms of cyclic graphs. This also allows us to circumvent further technicalities in checking both validity and cyclicity conditions (e.g., the introduction of so-called \emph{order-coherece} in~\cite{Graham-cyclic-first-order-mu-calculus}.).
\end{rem}

The main result of this section is soundness for reset proofs. This is a standard argument towards contradiction that leverages  a preliminary lemma, which shows that inference rules locally preserve soundness, and constructs an infinite branch of the reset proof that reflects falsity.  Specifically, our argument is essentially borrowed from ~\cite{Graham-cyclic-first-order-mu-calculus}.

\begin{lem}[Local soundness] \label{lemma:one-step-branch-lemma}

For any inference rule $\infrule$ of $\muPAord$ with conclusion $\constr: \Sigma$, if  $\model N , \rho \nvDash \constr : \Sigma $ then there is  a premise  $\constr' : \Sigma'$ of $\infrule$  and assignment $\rho'$ s.t.

\begin{enumerate}
\item \label{cond:1} $\model N , \rho' \nvDash \constr' : \Sigma' $ and
    \item \label{cond:2} for any $ \kappa \in \constr \cap \constr'$ we have
      $ \rho'(\kappa)\leq \rho(\kappa)$.
$\overline{ F(\overline{C}) }$
\end{enumerate}
If moreover, $\infrule= \rs \kappa$ then $\rho'(\kappa)<\rho(\kappa) $ \lukasedit{in~\autoref{cond:2}.}{.}
\end{lem}

\begin{proof}

We reason by case analysis on the the rule $\infrule$. We notice that $\infrule$ cannot have zero premises, as those rules are clearly sound in the standard model while we assume  $\model N , \rho \nvDash \constr : \Sigma$ by induction hypothesis. So, we suppose the rule has at least one premise. 

\begin{itemize}
    \item If the rule is $\vliinf{\cut}{}{\constr : \Gamma}{\constr :  \Gamma,  \phi}{\constr :  \Gamma, \neg \phi}$  then we define $\rho'$ to be the extension of $\rho$ that maps every fresh number variable in $\phi$ to $0$. Note that no fresh ordinal variables occur in $\phi$ since the constraint is unchanged. Then either $\model N , \rho' \nvDash  \phi$ or $\model N , \rho' \nvDash  \neg \phi$ and we set $\constr': \Gamma'$ to be the premise that is not satisfied. 
    \item  If the rule is $\vliinf{\wedge}{}{\constr : \Gamma, \phi \land \psi}{\constr :  \Gamma, \phi}{\constr :  \Gamma, \psi}$ then we set $\rho'\dfn \rho$. By the inductive hypothesis we have $\model N , \rho' \nvDash  \phi \land \psi$ and so $\model N , \rho' \nvDash  \phi$ or $\model N , \rho' \nvDash \psi $,  and we set $\constr': \Gamma'$ to be the premise that is not satisfied. 
    \item If the rule is $\vlinf{\vee}{}{\constr : \Gamma, \phi \lor \psi}{\constr : \Gamma, \phi,\psi}$ then we set $\rho'\dfn \rho$. By the induction hypothesis  we have $\model N , \rho' \nvDash  \phi \lor \psi$ and so $\model N , \rho'\nvDash  \phi, \psi$.
    \item If the rule is $\vlinf{\forall}{\text{\footnotesize $y$ fresh}}{\constr : \Gamma , \forall x \phi}{\constr : \Gamma, \phi[y/x]}$. By inductive hypothesis we have $\model N , \rho \nvDash  \forall x \phi $ and so there exists $a \in \mathbb{N}$ s.t. $\model N , \rho[x \mapsto a] \nvDash  \phi$. Since $y$ is fresh we can assign $y$ to $a$ without affecting $\Gamma$. We set  $\rho' = \rho[y \mapsto a]$. 
     \item If the rule is $\vlinf{\exists}{}{\constr : \Gamma , \exists x \phi(x)}{\constr : \Gamma, \phi(t)}$ then let $\rho'$ be the extension of $\rho$ that maps every fresh variable in $t$ to $0$. By induction hypothesis we have $\model N , \rho \nvDash  \exists x \phi $ and so for every $a \in \model N$ we have $\model N , \rho[x \mapsto a] \nvDash  \phi$, in particular for $a = \rho'(t)$.
    \item If the rule is $\vlinf{ \mu}{}{ \seqcon : \Gamma, t \in \mu Xx \phi}{\seqcon : \Gamma,  \phi[\mu X  x \phi/X, t/x]}$ then by induction hypothesis $\model N , \rho \nvDash t \in \mu Xx \phi $ and so $\rho(t)$ is not in the least fixed point $\mu f$ for $f:A \mapsto \{ a \in \model N \ | \  \model N , \rho \vDash \phi(A, a)   \}$. Since $\mu f$ is a fixed point, we have $\rho(t) \in \mu f \iff \rho(t) \in \{ a \in \model N \ | \  \model N , \rho \vDash \phi( \mu f , a)   \}$. Finally, note that $\rho(t) \in \{ a \in \model N \ | \  \model N , \rho \vDash \phi(\mu f , a)   \}$ if and only if $\model N , \rho \vDash \phi(\mu X  x \phi, t)$. Thus, we can simply let $\rho'\dfn\rho$.
    
    \item If the rule is $\vlinf{ \nu}{\text{\footnotesize $\lambda$ fresh}}{\seqcon : \Gamma, t \in \nu X x \phi  }{\seqcon + \lambda:  \Gamma, t \in \nu^{\lambda} X x \phi } $ then by induction hypothesis we have that $\model N , \rho \nvDash t \in \nu Xx \phi $ and so $\rho(t)$ is not in the greatest fixed point $\nu f$ of $f: A \mapsto \{ a \in \model N \ | \  \model N , \rho \vDash \phi(A, a)   \}$. Since $\rho(t) \notin \nu f$ there must be some ordinal $\alpha$ s.t. $\rho(t) \notin \nu^{\alpha} f$. Let $\rho' = \rho[\lambda \mapsto \alpha] $. Since $\lambda$ is fresh $\rho'$ will still falsify $\Gamma$ and it is clear that $\rho' \vDash \constr + \lambda$. Note that the function $f': A \mapsto \{ a \in \model N \ | \  \model N , \rho'\vDash \phi(A, a)   \}$ is identical to $f$ since $\lambda$ is fresh and so $\model N , \rho' \nvDash t \in \nu^{\lambda} Xx \phi $ since $\rho'(t)=\rho(t) \notin \nu ^{\alpha} f = \nu ^{\rho'(\lambda) } f'$.
\item If the rule is $\vlinf{ \nu^{\kappa}}{\text{\footnotesize $\lambda$ fresh}}{\seqcon : \Gamma, t \in \nu^{\kappa} X x \phi  }{\seqcon +_{\kappa} \lambda : \Gamma, \phi[\nu^{\lambda} X  x \phi/X, t/x]}$ then let $f:A \mapsto \{ a \in \model N \ | \  \model N , \rho \vDash \phi(A, a)   \}$. Since $\model N , \rho \nvDash t \in \nu^{\kappa} Xx \phi $ we have $\rho(t) \notin \nu ^{\rho(\kappa)} f$ and so $\rho(t) \notin f ( \nu ^{\alpha} f )$ for some $\alpha < \rho(\kappa)$. Let $\rho'= \rho[\lambda \mapsto \alpha]$. Note $\rho' \vDash \constr +_{\kappa} \lambda$, and since $\lambda$ is fresh $\Gamma$ is still falsified by $\rho'$,  and the function $f': A \mapsto \{ a \in \model N \ | \  \model N , \rho' \vDash \phi(A, a)   \}$ is identical to $f$. Now $\rho'(t)= \rho(t) \notin f ( \nu ^{\alpha} f ) = f' ( \nu ^{\rho' ( \lambda ) } f ) $ and so $\model N , \rho'\nvDash \phi[\nu^{\lambda} X  x \phi/X, t/x]$.

\item For any instance of $\mathsf{cw}$ that is not the reset rule (including the pruning rule), we simply let $\rho'\dfn \rho$.

    \item Finally, if the rule is 
$
\vlinf{\rs \kappa}{}{\constr: \Gamma}{\constr \setminus C(\kappa): \Gamma}
$
 then let the children of
$\kappa$ be $\lambda_0, \hdots \lambda_m$. Then we define $\rho'$ to be the restriction of $\rho$ to the remaining ordinal variables except $\rho'(\kappa) = \max \{ \rho(\lambda_0), \hdots \rho(\lambda_m) \}$. Since $\kappa$ does not occur in $\Gamma$ all we have to check to show that falsity is preserved is that $\rho' \vDash \constr'$. 
If $\constr' \proves \xi < \xi'$ for $\xi \neq \kappa$ and $ \xi' \neq \kappa$ then $\constr \proves \xi < \xi'$ and $\rho'(\xi) = \rho(\xi) < \rho(\xi')=\rho'(\xi')$. 
If $\constr' \proves \kappa < \xi$ then also $\constr \proves \kappa < \xi$ and so $\rho(\kappa) < \rho( \xi ) = \rho'( \xi ) $. But $\rho'(\kappa) = \max \{ \rho(\lambda_0), \hdots \rho(\lambda_m) \} < \rho(\kappa)$ and so $\rho'(\kappa) < \rho'( \xi )$
If $\constr' \proves \xi < \kappa$ then $\constr \proves  \xi < \lambda_i$ for some $i \leq m$ and so $\rho'(\xi) = \rho(\xi) < \rho( \lambda_i) \leq \max \{ \rho(\lambda_0), \hdots \rho(\lambda_m) \} = \rho' (\kappa)$.

    \item All the remaining cases (i.e., $\wk$, $\exch$, $\theta\text{-}\mathsf{sub}$ and $=\text{-}\mathsf{sub}$) are straightforward.
\end{itemize}

\end{proof}

\begin{thm}[Soundness] \label{thm:Graham-soundness}
    For any  reset proof $\der$  with conclusion $\constr : \Gamma$, $\N \models \constr: \Gamma$.
\end{thm}
\begin{proof}
Assume towards contradiction that $\der$ is a  reset proof with conclusion $\constr : \Gamma $ s.t.  $\model N , \rho \nvDash \constr : \Gamma$ for some assignment $\rho$. By \autoref{lemma:one-step-branch-lemma} we can construct  an infinite branch $(\constr_i : \Gamma_i )_{i \geq 0}$  of  $\der$   and  a sequence of assignments $(\rho_i)_{i < \omega}$  such that  $\constr_0 : \Gamma_0 \dfn \constr : \Gamma$, $\rho_0\dfn \rho$ and for all $i \geq 0$:
\begin{enumerate}
    \item \label{enum:sat1}$\model N , \rho_i \nvDash \constr_i : \Gamma_i$ 
    \item \label{enum:sat2} for any  $\kappa \in \constr_i\cap \constr_{i+1}$ we have  $\rho_{i+1} (\kappa) \leq \rho_i(\kappa)$
    \item \label{enum:sat3} if $\kappa \in \constr_i$ and $\constr_i : \Gamma_i$ is the conclusion of a reset rule $RS(\kappa)$  then $\rho_{i+1} (\kappa) < \rho_i(\kappa)$.
\end{enumerate}

Since $\der$ is a reset proof, there will be some stable reset variable $\kappa$ along this branch, and so for some $m$ we  have $\kappa \in \constr_{n}$ for every $n\geq m$. Consider the ordinal sequence $(\rho_n (\kappa) )_{n \geq m}$. By construction we have that $\rho_{n+1} (\kappa) \leq \rho_n (\kappa)$ for every $n \geq m$ and $\rho_{n+1} (\kappa) < \rho_n (\kappa)$ at every reset step. Since $\kappa$ is reset infinitely many times this contradicts wellfoundedness of ordinals. 
\end{proof}

\section{From progressing proofs to reset proofs}\label{sec 5}

In this section we define a translation of ($\muPA$-)coderivations of a sequent $\Gamma$ into annotated coderivations \emph{of the same sequent} (endowed with an empty constraint). The main result of this section is  that, in particular, the translation maps progressing proofs into reset proofs and  preserves regularity of the proof tree. This allows us to obtain soundness of progressing proofs from~\autoref{thm:Graham-soundness}.

\subsection{The translation}

Our translation will essentially ``decorate'' $\muPA$-coderivations with ordinal variables and corecursively construct constraints accordingly.

\begin{defi}[Ordinal assignments]
An \emph{ordinal assignment} is a list $\ell$ over $\mathcal{OV}\cup\{\star\}$. We denote concatenation of two ordinal assignments $\ell$ and $\ell'$ with $\ell \cdot \ell'$. Given a formula $\phi$ of  $\muPA$, we say that \emph{$\ell$ has type $\phi$}, written $\ell\dect \phi$, if the number of occurrences of $\nu$ in $\phi$ is equal to the length of $\ell$. Given $\ell\dect\phi$ with $\ell=a_1, \ldots, a_k$, the   \emph{$\ell$-annotation of $\phi$},  written $\phi^\ell$, is  a formula of $\langmuarithord$ obtained by replacing the  $i$-th occurrence of $\nu$ with $\nu^{a_i}$. We call  $\vec{\star}= \star \overset{k}{\ldots} \star$ the \emph{identity ordinal assignment (for $\phi$)}, and $\phi^{\vec \star}$ its  \emph{identity annotation}. 

{Given $\ell \dect \phi(x, X)$ we denote with $\nu^\kappa(\ell)$ the \emph{unique} ordinal assignment $\nu^\kappa(\ell)\dect \phi(t, \nu Xx.\phi)$ such that  $(\phi(t, \nu Xx.\phi))^{\nu^\kappa(\ell)}=\phi^\ell(t, \nu^\kappa Xx.\phi^\ell)$}. 
{Similarly, given $\ell \dect \phi(x, X)$ we denote with $\mu(\ell)$ the \emph{unique} ordinal assignment  $\mu(\ell)\dect \phi(t, \mu Xx.\phi)$ such that  $(\phi(t, \mu Xx.\phi))^{\mu(\ell)}=\phi^\ell(t, \mu Xx.\phi^\ell)$.
}

We generalise ordinal assignments to contexts. Given $\Gamma= \phi_1, \ldots, \phi_n$ with $\ell_i \dect\phi_i$, we say that \emph{$\vec \ell= \ell_1, \ldots, \ell_n$ has type $\Gamma$}, written  $\vec  \ell\dect \Gamma$, if  $\ell_1\dect \phi_1, \ldots, \ell_n\dect \phi_n$. Similarly, the  \emph{$\vec \ell$-annotation of $\Gamma$}, written $\Gamma^{\vec \ell}$, is the list  $\phi_1^{\ell_1}, \ldots, \phi_n^{\ell_n}$.

\end{defi}

\begin{exa}
    Consider $\phi\dfn t \in \nu X x (s \in \nu Yy (x \in X \vee y \in Y \vee \ntype))$, and let $\kappa\in \mathcal{OV}$. Then, $\ell\dfn\kappa \star$ is an ordinal assignment for $\phi$, and $\phi^\ell=t \in \nu^\kappa X x (s \in \nu^\star Yy (x \in X \vee y \in Y \vee \ntype))=t \in \nu^\kappa X x (s \in \nu Yy (x \in X \vee y \in Y \vee \ntype))$.
\end{exa}

\begin{figure}
    \centering
     \[
\begin{array}{rcl}
    \begin{pmatrix}
         \vlderivation{
\vlin{}{\infrule}{\Gamma}{\vlhy{}}
    }
    \end{pmatrix}^{\constr\mid \vec \varepsilon}  &\mapsto &
   \vlderivation{
   \vlin{}{\infrule}{\constr: \Gamma}{\vlhy{}}
   } \qquad \infrule \in \{\id, =,   \mathsf{PA}_1, \ldots, \mathsf{PA}_7, \ntype\}
\\ \\ \\
 \begin{pmatrix}
         \vlderivation{
\vlin{}{\wk}{ \Gamma, \Delta}{\vldr{\der}{\Gamma}}
    }
    \end{pmatrix}^{\constr\mid \vec \ell, \vec{\ell'}}  &\mapsto &
    \vliqf{}{\rp}{\constr:  \Gamma^{\vec \ell}, \Delta^{\vec {\ell'}}}{
 \vlinf{}{\wk}{\rpf\constr:   \Gamma^{\vec  \ell}, \Delta^{\vec {\ell'}}}{\ \ \ \ \begin{pmatrix} \vltreeder{\der}{ \Gamma}{\ \  }{\ \ \ \ }{\ \  }\end{pmatrix}^{\rpf \constr\mid  \vec \ell}}
    } \\ \\ \\ 
     \begin{pmatrix}
         \vlderivation{
\vlin{}{\exch}{ \Gamma, \phi , \psi, \Delta}{\vldr{\der}{\Gamma, \psi, \phi, \Delta}}
    }
    \end{pmatrix}^{\constr\mid \vec \ell, \ell , \ell', \vec{\ell'}}  \!\!\!\! &\mapsto  &
    \vliqf{}{\rp}{\constr: \Gamma^{\vec \ell}, \phi^{\ell}, \psi^{\ell'}, \Delta^{\vec {\ell'}}}{
 \vlinf{}{\exch}{\rpf\constr:   \Gamma^{\vec \ell}, \phi^{\ell}, \psi^{\ell'}, \Delta^{\vec {\ell'}}}{\ \ \ \ \begin{pmatrix} \vltreeder{\der}{ \Gamma, \phi, \psi, \Delta}{\ \  }{\ \ \ \ }{\ \  }\end{pmatrix}^{\rpf \constr\mid  \vec \ell, \ell', \ell, \vec{\ell'}}}
    } 
\\ \\ \\ 
    \begin{pmatrix}
         \vlderivation{
\vlin{}{\infrule}{ \Gamma}{\vldr{\der}{\Gamma'}}
    }
    \end{pmatrix}^{\constr\mid \vec \ell}  &\mapsto  &
    \vliqf{}{\rp}{\constr:  \Gamma^{\vec \ell}}{
 \vlinf{}{\infrule}{\rpf\constr:  \Gamma^{\vec \ell}}{\ \ \ \ \begin{pmatrix} \vltreeder{\der}{ \Gamma'}{\ \  }{\ \ \ \ }{\ \  }\end{pmatrix}^{\rpf \constr\mid \vec \ell}}
    }\qquad\infrule\in \{\theta \text{-} \mathsf{sub},\forall, \exists,  = \text{-} \mathsf{sub}\} 
  \\ \\ 
    \end{array} 
\]
       \caption{Translating coderivations to annotated coderivations (Part 1).}
    \label{fig:transl-1}
\end{figure}

\begin{figure}
    \centering
    \small  \[
\begin{array}{rcll}
      \begin{pmatrix}
         \vlderivation{
\vlin{}{\vee}{ \Gamma, \phi \vee \psi}{\vldr{\der}{\Gamma, \phi, \psi}}
    }
    \end{pmatrix}^{\constr\mid \vec \ell, \ell \cdot \ell'} & \mapsto & 
    \vliqf{}{\rp}{\constr:  \Gamma^{\vec \ell}, \phi^\ell \vee \psi^{\ell'} }{
 \vlinf{}{\vee}{\rpf\constr:  \Gamma^{\vec \ell}, \phi^\ell \vee \psi^{\ell'}}{\ \ \ \ \begin{pmatrix} \vltreeder{\der}{ \Gamma, \phi, \psi}{\ \  }{\ \ \ \ }{\ \  }\end{pmatrix}^{\rpf \constr\mid \vec \ell, \ell, \ell'}}
    }& \\ \\ 
     \end{array} 
     \]
    \[
\begin{array}{rcl}
    \begin{pmatrix}
         \vlderivation{
\vliin{}{\wedge}{\Gamma, \phi \wedge \psi}{\vldr{\der}{\Gamma, \phi}}{\vldr{\der'}{\Gamma, \psi}}
    }
    \end{pmatrix}^{\constr\mid \vec \ell, \ell\cdot\ell'} & \mapsto & 
    \vliqf{}{\rp}{\constr: \Gamma^{\vec \ell}, \phi^\ell \wedge \psi^{\ell'}}{
 \vliinf{}{\wedge}{\rpf\constr: \Gamma^{\vec \ell}, \phi^\ell \wedge \psi^{\ell'}}{\ \ \ \ \begin{pmatrix} \vltreeder{\der}{\Gamma, \phi}{\ \  }{\ \ \ \ }{\ \  }\end{pmatrix}^{\rpf \constr\mid  
 \vec \ell, \ell}}{\ \ \ \ \begin{pmatrix} \vltreeder{\der'}{ \Gamma,  \psi}{\ \  }{\ \ \ \ }{\ \  }\end{pmatrix}^{\rpf \constr\mid\vec \ell, \ell'}}
    }
\\ \\ \\ 
    \begin{pmatrix}
         \vlderivation{
\vliin{}{\cut}{\Gamma}{\vldr{\der}{\Gamma, \phi}}{\vldr{\der'}{\Gamma, \neg\phi}}
    }
    \end{pmatrix}^{\constr\mid \vec \ell} & \mapsto & 
    \vliqf{}{\rp}{\constr: \Gamma^{\vec \ell}}{
 \vliinf{}{\cut}{\rpf\constr: \Gamma^{\vec \ell}}{\ \ \ \ \begin{pmatrix} \vltreeder{\der}{ \Gamma, \phi}{\ \  }{\ \ \ \ }{\ \  }\end{pmatrix}^{\rpf \constr\mid \vec \ell, \vec \star}}{\ \ \ \ \begin{pmatrix} \vltreeder{\der'}{ \Gamma, \neg \phi}{\ \  }{\ \ \ \ }{\ \  }\end{pmatrix}^{\rpf \constr\mid \vec \ell, \vec \star}}
    }
\end{array}
   \]
    \caption{Translating coderivations to annotated coderivations (Part 2).}
    \label{fig:transl-2}
\end{figure}

\begin{figure}
    \centering
    \[
\begin{array}{rcll}
 \begin{pmatrix}
         \vlderivation{
\vlin{}{\mu}{ \Gamma, t \in \mu Xx.\phi}{\vldr{\der}{\Gamma, \phi(t, \mu Xx. \phi)}}
    }
    \end{pmatrix}^{\constr\mid \vec \ell, \ell } & \mapsto & 
    \vliqf{}{\rp}{\constr: \Gamma^{\vec \ell}, t \in \mu Xx.\phi^{\ell}}{
    \vlinf{}{\mu}{\rpf\constr:   \Gamma^{\vec \ell}, t \in \mu Xx.\phi^{\ell}}{
    \vlidf{}{=}{\rpf\constr:   \Gamma^{\vec \ell}, \phi^\ell (t, \mu Xx.\phi^\ell)}{
 \vlidf{}{=}{\rpf\constr:   \Gamma^{\vec \ell}, (\phi (t, \mu Xx.\phi))^{\mu(\ell)}}{\ \ \ \ \begin{pmatrix} \vltreeder{\der}{ \Gamma, \phi (t, \mu Xx.\phi)}{\ \  }{\ \ \ \ }{\ \  }\end{pmatrix}^{\rpf \constr\mid  \vec \ell, \mu(\ell)}}
    }}} \\ \\ \\ 
     \begin{pmatrix}
         \vlderivation{
\vlin{}{\nu}{ \Gamma, t \in \nu Xx.\phi}{\vldr{\der}{\Gamma,\phi(t, \nu Xx.\phi)}}
    }
    \end{pmatrix}^{\constr\mid \vec \ell, \star\ell } & \mapsto & 
    \vliqf{}{\rp}{\constr: \Gamma^{\vec \ell}, t \in \nu Xx.\phi^{\ell}}{
 \vlinf{}{\nu}{\rpf\constr:   \Gamma^{\vec \ell}, t \in \nu Xx.\phi^{\ell}}{\vlinf{}{\nu^\kappa}{\rpf\constr+ \lambda: \Gamma^{\vec \ell},t \in \nu^\lambda Xx.\phi^{\ell} }{\vlidf{}{=}{(\rpf\constr+ \lambda) +_\lambda \kappa: \Gamma^{\vec \ell},
\phi^\ell(t, \nu^\kappa Xx.\phi^{ \ell}) }{\vlidf{}{=}{(\rpf\constr+ \lambda) +_\lambda \kappa: \Gamma^{\vec \ell},(\phi(t,\nu Xx.\phi))^{\nu^\kappa( \ell)}}{ {\ \ \ \ \begin{pmatrix} \vltreeder{\der}{ \Gamma, \phi(t, \nu Xx.\phi)}{\ \  }{\ \ \ \ }{\ \  }\end{pmatrix}^{(\rpf\constr+ \lambda) +_\lambda \kappa\mid  \vec \ell, \nu^\kappa( \ell)}}}}}}
    }  \\ \\ \\ 
         \begin{pmatrix}
         \vlderivation{
\vlin{}{\nu}{ \Gamma, t \in \nu Xx.\phi}{\vldr{\der}{\Gamma, \phi(t, \nu Xx.\phi)}}
    }
    \end{pmatrix}^{\constr\mid \vec \ell, \lambda\ell } & \mapsto & 
    \vliqf{}{\rp}{\constr: \Gamma^{\vec \ell}, t \in \nu^\lambda Xx.\phi^{\ell}}{\vlinf{}{\nu^\kappa}{\rpf\constr: \Gamma^{\vec \ell},t \in \nu^\lambda Xx.\phi^{\ell} }{\vlidf{}{=}{\rpf\constr +_\lambda \kappa: \Gamma^{\vec \ell},
\phi^\ell(t, \nu^\kappa Xx.\phi^{ \ell}) }{\vlidf{}{=}{\rpf\constr +_\lambda \kappa: \Gamma^{\vec \ell},(\phi(t,\nu Xx.\phi))^{\nu^\kappa( \ell)}}{ {\ \ \ \ \begin{pmatrix} \vltreeder{\der}{ \Gamma, \phi(t, \nu Xx.\phi)}{\ \  }{\ \ \ \ }{\ \  }\end{pmatrix}^{\rpf\constr +_\lambda \kappa\mid  \vec \ell, \nu^\kappa( \ell)}}}}}
    }  
     \end{array} 
    \]
    % \todoinline{Gianluca: $\phi^\ell$ both with free variable and with nu-term does not work well.}
    \caption{Translating coderivations to annotated coderivations (Part 3).}
    \label{fig:transl-3}
\end{figure}

\begin{defi}[Translation]\label{defn:translation}  Let $\constr$ and $\vec \ell$ such that $\vec \ell\dect \Sigma$ {and all ordinal variables in $\ell$ occur in $\constr$}. We define coinductively the translation $(\cdot)^{\constr\mid\vec \ell}$  mapping a $\muPA$-coderivation  with conclusion $\Sigma$  into an annotated coderivation  with conclusion $\constr : \Sigma^{\vec \ell}$ as in~\autoref{fig:transl-1},  \autoref{fig:transl-2} and \autoref{fig:transl-3} and relying on the following conventions:
\begin{enumerate}
\item W.l.o.g., to simplify the translation we assume that $\id$ steps in $\muPA$-coderivations have only atomic conclusions (see \autoref{prop:gen-id} in~\autoref{app:cmupa}).
    \item \label{conv1} We fix an enumeration $\xi_0, \xi_1, \ldots$ of  $\mathcal{OV}$ and  assume that, when constructing an annotated coderivation bottom-up, the rules $\nu$ and  $\nu^\kappa$  pick from the enumeration the ordinal variable with least index  that does not appear in $\constr$. 
    \item \label{conv2}  {With 
    $\vliqf{\rp}{}{\constr: \Gamma}{\constr^\rp: \Gamma}$
we denote that $\constr^\rp: \Gamma$ is obtained from $\constr: \Gamma$ by applying $RS(\kappa)$ and $\pr$ until no instances of $RS(\kappa)$ or $\pr$ can be applied to $\constr^\rp$. Note that this is possible since $\constr$ is finite and each application of $RS(\kappa)$ and $\pr$ removes at least one element.}
\end{enumerate}

 \end{defi}

We conclude by noticing that the translation essentially preserves the structure of branches as well as their progressing traces.

\begin{fact}\label{lem:branch-to-branch}
  Let $\der$ be a $\muPA$-coderivation. Any branch
  $\branch=(\constr_i: \Gamma_i)_{i \geq 0}$ of
  $\der^{\emptyset \mid \vec{\star}}$ induces a unique branch
  $\widehat{\branch}=(\Delta_i)_{i \geq 0}$ of $\der$ such that
  $(\Gamma_i^{-})_{i \geq 0}$ enumerates $(\Delta_i)_{i \geq 0}$, possibly with repetitions.

  Moreover, if $\widehat{\tau}= (\phi_i)_{i \geq i_0}$ is a progressing trace of
  $\widehat{\branch}$, and $\nu Xx.\phi$ is its smallest infinitely progressing set-term,
  then there is a trace $\tau=(\psi_i)_{i \geq i_1}$ of $\branch$ such that
  $\tau^- \dfn (\psi_i^-)_{i \geq i_1}$ enumerates $(\phi_i)_{i \geq i_0}$, possibly with
  repetitions. Hence $\tau^-$ is a progressing trace and $\nu Xx.\phi$ is the smallest
  infinitely progressing set-term in $\tau^-$.
\end{fact}

\begin{fact}\label{lem:branch-to-branch}
  Let $\der$ be a $\muPA$-coderivation.  Any branch $\branch=(\constr_i: \Gamma_i)_{i \geq 0}$ of $\der^{\emptyset \mid \vec{\star}}$ induces a unique 
    branch $\widehat{\branch}=(\Delta_i)_{i \geq 0}$ of $\der$ such that  $(\Gamma_i^{-})_{i \geq 0}$ enumerates  $(\Delta_i)_{i \geq 0}$ (possibly with repetitions). Moreover, if $\widehat{\tau}= (\phi_i)_{i \geq i_0}$ is a progressing trace of $\widehat{\branch}$, and $\nu Xx.\phi$ is its smallest infinitely set term, then there is  $\tau=(\psi_i)_{i \geq i_1}$ of $\branch$ such that
         $\tau^-\dfn(\psi^-_i)_{i \geq i_1}$ enumerates $(\phi_i)_{i \geq i_0}$ (possibly with  repetitions), so 
     $\tau^-$ is a progressing trace and $\nu Xx.\phi$ is the smallest infinitely set-term in $\tau^-$. 
\end{fact}

\subsection{Progressivity implies reset condition}

As it stands, the translation of a progressing proof might not be a reset proof. To see this, consider the following example:
\[
\der\dfn 
\begin{array}{rcl}
     \vlderivation{
\vliq{}{\nu, \vee}{\blue{\phi}}{
  \vliq{}{\nu, \vee}{\phi, \blue{\phi}}{
   \vliq{}{\nu, \vee}{\phi, \phi, \blue{\phi}}{
      \vliq{}{\nu,, \vee}{\phi, \phi, \phi, \blue{\phi}}{\vlhy{\vdots}}
  }
  }
}
}
\quad 
&\mapsto& \quad 
 \vlderivation{
\vliq{}{\nu, \vee}{\emptyset:\phi}{
\vliq{}{\nu^{\kappa_1}, \vee}{\{\lambda\}:\phi^\lambda}{
  \vliq{}{\nu^{\kappa_2}}{\{\kappa_1<\lambda\}:\phi^{\kappa_1}, \phi^{\kappa_1}}{
      \vliq{}{\nu^{\kappa_3}, \vee}{\{\kappa_2<\kappa_1<\lambda\}:\phi^{\kappa_1}, \phi^{\kappa_2}, \phi^{\kappa_2}}{
      \vliq{}{\nu^{\kappa_4}, \vee}{\{\kappa_3<\kappa_2<\kappa_1<\lambda\}:\phi^{\kappa_1}, \phi^{\kappa_2}, \phi^{\kappa_3}, \phi^{\kappa_3}}{\vlhy{\vdots}  }
      }
  }
  }
}
}
=\der^{\emptyset\mid\vec \star}
\end{array}
\]
where $\phi \dfn t \in \nu X x (x \in X \vee x \in X)$, $t \in \phi^\kappa \dfn \nu^\kappa X x (x \in X \vee x \in X)$, the principal formulas are depicted in $\blue{blue}$,  and double lines mean multiple rule steps. It is easy to see that  $\der$ is a progressing proof, as the $\blue{blue}$ formulas form a progressing trace along its only infinite branch. 
{However, even though the translation applies the reset rule whenever possible, $\der^{\emptyset\mid \vec \star}$ does not contain any reset step because all the ordinal variables in the constraint also appear in the sequent.
Notice, that this counterexample crucially relies on the fact that $\der$ has sequents with unbounded size.
}

The example above  justifies the restriction of  progressing proofs to those with \emph{finitely} many sequents in~\autoref{lem:preservation-of-finiteness} and~\autoref{lem:preservation-of-progressivity}. Clearly, this restriction can be dropped  when progressing proofs are cyclic, as it is a consequence of regularity of the prooftree.

We start with a preliminary lemma that will be used to bound the number of ordinal variables in the constraints of a translated coderivations.

 \begin{lem}\label{lem:size-bound-on-prune-and-reset-free-sequent-constraint}
 Let $\constr : \Gamma $ be an extended sequent that cannot be conclusion of a reset rule nor of a pruning rule. Then $2N \geq | \constr | $, where   $N$ is the number of ordinal variables in $\Gamma$.  
    % Let $\constr : \Gamma $ be such that no reset rule or pruning rule can be applied. Let $N$ be the number of ordinal variables in $\Gamma$. Then $2N \geq | \constr | $
\end{lem}

\begin{proof}
    We will define an injection $f : \{ \beta \in \constr \ \mid \ \beta \text{ does not occur in $\Gamma$} \} \rightarrow \{ \gamma \in \constr \ \mid \ \gamma \text{ does occur in $\Gamma$} \}$. From this it follows that $| \constr | = |\{ \beta \in \constr \ \mid \ \beta \text{ does not occur in $\Gamma$} \}| + | \{ \gamma \in \constr \ : \ \gamma \text{ does occur in $\Gamma$} \} | \leq N +N$. 
    
    Say $\beta \in \constr$ does not occur in $\Gamma$. Since $\constr : \Gamma $ cannot be conclusion of a  pruning rule, $\beta$ cannot be a leaf. But then,  since $\constr : \Gamma $ cannot be conclusion of a reset rule, at least one of the children of $\beta$ must occur in $\Gamma$. Let $\gamma$ be one of these, and let us set $f(\beta)\dfn \gamma$. Note that since $\constr$ is a forest we must have that $f(\beta) \neq f(\beta')$ when $\beta \neq \beta'$.  
\end{proof}

\begin{lem}\label{lem:preservation-of-finiteness}
     Let $\der$ be a progressing proof  with finitely many sequents. Then $\der^{\emptyset \mid \vec{\star}}$ has {only finitely many extended sequents.} 
\end{lem}
\begin{proof}
We first show that $\der^{\emptyset \mid \vec{\star}}$ contains only finitely many ordinal variables.  Since  $\der$ has only finitely many sequents, there is a bound on the number of set terms that occur in a sequent of $\der$, and so there is a bound on the number of ordinal variables that can occur in an annotated sequent of $\der^{\emptyset \mid \vec{\star}}$. Denote this last bound by $N$.    We now   show that each $\xi_m$ with $m \geq 2N+1$ cannot appear in $\der^{\emptyset \mid \vec{\star}}$. Assume towards contradiction that there is such a  $\xi_m$ in $\der^{\emptyset \mid \vec{\star}}$. By inspecting~\autoref{fig:mupaord rules-extended} there must be in $\der^{\emptyset \mid \vec{\star}}$ a rule step of one of the following forms: 
 \[
 \vliqf{\rp}{}{\seqcon : \Gamma, t \in \nu X x \phi }{
    \vlinf{ \nu}{}{\seqcon^\rp : \Gamma, t \in \nu X x \phi  }{\seqcon^\rp + \xi_m:  \Gamma, t \in \nu^{\xi_m} X x \phi } 
    }
    \qquad 
    \vliqf{\rp}{}{\seqcon : \Gamma, t \in \nu^\kappa X x \phi }{
\vlinf{ \nu^{\kappa}}{}{\seqcon^\rp : \Gamma, t \in \nu^{\kappa} X x \phi  }{\seqcon +_{\kappa} \xi_m : \Gamma, \phi[\nu^{\xi_m} X  x \phi/X, t/x]} 
}
\]
and so $\xi_m\not \in \constr^\rp$. Now, by definition of the translation,  the rules $\nu$ and  $\nu^\kappa$  pick from the enumeration $\xi_0 , \xi_1 \hdots $ the ordinal variable with least index that does not appear in the constraint (\autoref{defn:translation}.\ref{conv1}). This means that $\constr^\rp$ must contain every $\xi_i$ for $i < m$ and so $| \constr^\rp | \geq m \geq 2N+1$. On the other hand, $\constr^\rp$ cannot be in the conclusion of a reset or a pruning rule (\autoref{defn:translation}.\ref{conv2}). By~\autoref{lem:size-bound-on-prune-and-reset-free-sequent-constraint} this implies $2N \geq | \constr^\rp |$. A contradiction. 

We now show that $\der^{\emptyset \mid \vec{\star}}$ has also finitely many extended sequents. Since it contains only finitely many ordinal variables, by combinatorics it also contains only finitely many constraints. Moreover, since every extended sequent $\constr: \Gamma$ of $\der^{\emptyset \mid \vec{\star}}$  is such that $\Gamma^-$ is a sequent of $\der$, which are only finitely many by assumption,  and every ordinal variable occurring in the annotated sequent $\Gamma$ must occur in $\constr$, we conclude that there can be only finitely many annotated (hence extended) sequents in $\der^{\emptyset \mid \vec{\star}}$.
\end{proof}

\begin{lem}\label{lem:preservation-of-progressivity}
    Let $\der$ be an progressing proof  with finitely many sequents. Then  $\der^{\emptyset \mid \vec{\star}}$ is a  reset proof. 
\end{lem}
\begin{proof}
 Let $\mathcal{D}$ be a $\cmuPA$ coderivation with conclusion $\phi$. Let $\branch= (\constr_i: \Gamma_i)_{i \geq 0}$ be an infinite branch in $\der^{\emptyset \mid \vec{\star}}$. By~\autoref{lem:branch-to-branch} there exists a unique infinite branch $\widehat{\branch}= (\Delta_i)_{i \geq 0}$ of $\der$ that translates into $\branch$. By hypothesis, $\widehat{\branch}$ has a progressing trace $\widehat \tau=(\phi_i)_{i \geq i_0}$ for some $i_0\geq 0$ and a $\nu$-term $\nu X x. \phi$ that is the smallest infinitely progressing set-term in $\widehat{\tau}$. Again, by~\autoref{lem:branch-to-branch}   there is a progressing trace $\tau=(\psi_i)_{i \geq i_1}$ of $\branch$ such that $\tau^-\dfn (\psi_i^-)_{i \geq i_1}$ enumerates $(\phi_i)_{i \geq i_0}$ (possibly with  repetitions), and $\nu Xx.\phi$ is the smallest infinitely progressing set-term in $\tau^-$.  {Now, since $\tau^-$ is a progressing trace (and by~\autoref{prop:equivalent-trace-condition})}, there is $i_2 \geq i_1$ and a  sequence of ordinal variables $(\kappa_i)_{i \geq i_2}$  such that for all $i \geq i_2$:

\begin{enumerate}[(i)]
\item  \label{enum:obs1} 
$\nu^{\kappa_i} Xx.\phi$ occurs in $\psi_i$.
\item \label{enum:obs2}  
 Either $\kappa_{i+1} = \kappa_i$ or $\kappa_{i+1}$ is the child of $ \kappa_i$ in $\constr_{i+1}$ (i.e., in any case, $\kappa_i \geq_{\constr_{i+1}} \kappa_{i+1}$).
\item \label{enum:obs3}
There exists $j \geq i$ s.t. $\kappa_{j+1}$ is the child of $\kappa_{j}$ in $\constr_{j+1}$.
\end{enumerate}
Note that  $K \dfn\{\kappa_i \mid i \geq i_2 \}$ is a finite set  by~\autoref{lem:preservation-of-finiteness}. 

Let us show that $\branch$ has a stable variable. We  set $A_i\dfn \{\lambda\in \constr_i \mid  \constr_i \vdash \lambda \geq \kappa_i\}$ for  $i \geq i_2$. By definition of constraint, $A_i$ is a (finite) linear order for every $i \geq i_2$. Moreover, $A_i$ is non-empty since $\kappa_i \in A_i$ by~\autoref{enum:obs1}. This means that there is  a maximal element $\theta \in A_{i_2}$. Note that $\theta$ is also maximal in $\constr_{i_2}$. {We now show by induction that $\theta$ occurs in every $A_i$ for $i \geq i_{i_2}$ and so $\theta$ is a stable  variable. Let $i \geq i_{2}$ and assume $\theta \in A_i$.  Since $\theta$ doesn't have a parent in $\constr_{i_2}$ and the rules along $\branch$ cannot introduce a parent of $\theta$, we see that $\theta$ can never be removed by a reset rule. Since $\theta \in A_i$ we have $\constr_i \vdash \theta \geq \kappa_i$ and so by ~\autoref{enum:obs1} $\theta$ can't be pruned either. Thus $\theta \in \constr_{i+1}$ and by ~\autoref{enum:obs2} we have $\theta \geq_{ \constr_{i+1} } \kappa_{i} \geq_{ \constr_{i+1} } \kappa_{i+1}$ and so $\theta \in A_{i+1}$. }

{Notice that since $K$ is finite we can pick $i_2$ large enough such that every stable variable of $\branch$ occurs in every $\constr_i$ with $i \geq i_2$. }
Now, let $\theta^*$ be a minimal stable variable in  $A_{i_2}$, which exists because $A_{i_2}$ is a finite linear order that contains a stable variable. We show that $\theta^*$ is also a reset variable, i.e.,  for every ${i_3} \geq i_2$ there is $l \geq {i_3}$ such that  $\rs {\theta^*}$ is the rule with conclusion $\constr_l:\Gamma_l$ in $\branch$. 

First, by~\autoref{enum:obs3} there exists $j \geq {i_3}$ s.t. $\kappa_{j+1}$ is the child of $\kappa_{j}$ in $\constr_{j+1}$.
{Note that, for all $i_3 \leq i \leq j+1$, since $\theta^*$ is stable and so can't be removed, if $\theta^* \in A_i$ then $\theta^* \in \constr_{i+1}$ and so by \autoref{enum:obs2} $\theta^* \geq_{ \constr_{i+1} } \kappa_{i} \geq_{ \constr_{i+1} } \kappa_{i+1}$ we have $\theta \in A_{i+1}$. By induction we conclude that $\theta^* \in A_{j}$ and so $\theta^* \geq_{\constr_{j}} \kappa_{j} $. Furthermore, since also $\theta^* \in A_{j+1}$, we have $\theta^* \geq_{\constr_{j+1}} \kappa_{j} >_{\constr_{j+1}} \kappa_{j+1}$ and so $\theta^*$ has a child $\lambda \in A_{j+1}$.} 

Since $\theta^* > \lambda$ we have by our choice of $i_2$ and minimality of $\theta^*$ that $\lambda$ cannot be stable. This means that there exists $l \geq j$ such that $\lambda$ is removed by a pruning  or a reset rule $\infrule$ with conclusion $\constr_{l}:\Gamma_l$ in $\branch$, and we consider the least such $l$.
{For all $j+1 \leq i < l$, if $\lambda \in A_i$ then by assumption $\lambda \in \constr_{i+1}$ and so by \autoref{enum:obs2} $\lambda \geq_{ \constr_{i+1} } \kappa_{i} \geq_{ \constr_{i+1} } \kappa_{i+1}$ we have $\theta \in A_{i+1}$. Thus, by induction, we get $\lambda\in A_{l}$. By ~\autoref{enum:obs1} and $\lambda \geq_{\constr_l} \kappa_l$ we see that $\infrule$ cannot be the pruning rule, which means $\infrule= \rs{\chi}$ for some $\chi$. Finally, since $ \theta^* $ is stable we have that $ \theta^* \in \constr_l $. Since no rule along $ \branch $ can introduce a parent of an ordinal variable that already occurs in a constraint, we see that $ \theta^* $ is the parent of $ \lambda $ in $ \constr_l $. Moreover, since parents are unique, it must be that $ \chi=\theta^* $.}

\end{proof}

We can now prove the fundamental result of this section:

\begin{thm}\label{thm:translation} If $\cmuPA \vdash \phi$ then $\cmuPAord \vdash \emptyset:\phi$.  
\end{thm}
\begin{proof}

Let $\der$ be a cyclic progressing proof of $\phi$. Then the translation $\der^{\emptyset \mid \vec \star}$ proves $\emptyset: \phi^{\vec \star}$, i.e. $\emptyset: \phi$. Since a cyclic progressing proof has finitely many sequents, then by~\autoref{lem:preservation-of-progressivity} $\der^{\emptyset \mid \vec \star}$ is a reset proof with finitely many extended sequents.  

We now prove that  $\der^{\emptyset \mid \vec \star}$ is cyclic. Let $\branch=(v_j)_{j \geq 0}$ be an infinite branch of $\der^{\emptyset \mid \vec \star}$ consisting of extended sequents $( \constr_j : {\Gamma_j}^{\ell_{j}} )_{j \geq 0}$. By inspecting the translation, there is a (unique) branch $\branch'=(u_i)_{i \geq 0}$ of $\der$ s.t. for some strictly monotone $f: \mathbb{N}\to \mathbb{N}$ we have

$$ \der_{u_i}^{ \constr_{f(i)} \mid \ell_{f(i)}  }  = (\der^{\emptyset \mid \vec \star})_{v_{f(i)}} .$$

Since $\der$ is regular we can find an infinite $I \subseteq \Nat$ s.t. $\der_{u_i} = \der_{u_{i'}}$ for $i,i' \in I$. 

By~\autoref{lem:preservation-of-progressivity} there are only finitely many distinct extended sequents in $\branch$ and so there can only be finitely many distinct $\constr_j$ and $\ell_j$. Thus we can find $i,i' \in I$ with $i \neq i'$ and s.t. $\constr_i =  \constr_{i'}$ and $\ell_i = \ell_{i'}$. But then

$$ (\der^{\emptyset \mid \vec \star})_{v_{f(i)}} =
\der_{u_i}^{ \constr_{f(i)} \mid \ell_{f(i)}  }  =  \der_{u_{i'}}^{ \constr_{f(i')} \mid \ell_{f(i')}  }  = (\der^{\emptyset \mid \vec \star})_{v_{f(i')}} $$

Since $i \neq i'$ and $f$ is strictly monotone we have $f(i) \neq f(i')$ and so each infinite branch $\branch=(v_j)_j$ of $\der^{\emptyset\mid \vec \star}$ has two distinct nodes that are roots of identical sub-derivations, i.e. $\der^{\emptyset\mid \vec \star}$ is cyclic.

\end{proof}

Soundness for progressing proofs follows as a straightforward corollary:

\begin{cor}[Soundness for $\cmuPA$]\label{thm:soundness-nwmuPA}
    Let $\der\in \cmuPA$ with conclusion $\Gamma$. Then,  $\N \models \Gamma$. 
\end{cor}
\begin{proof}
    By~\autoref{thm:translation} we have that $\der^{\emptyset \mid \vec \star}$ is a reset proof of $\emptyset: \Gamma^{\vec \star}=\Gamma$. By~\autoref{thm:Graham-soundness} we have $\N \models \emptyset: \Gamma$, and so $\N \models \Gamma$.
\end{proof}

\section{Formalising the soundness argument in $\PSCA$}\label{sec 6}

The final section of this paper is devoted to proving the converse of \autoref{id-to-cid}.
To this end,  taking  inspiration from previous work \cite{Simpson17:cyc-arith,Das20:ca-log-comp}, we use second-order theories to formalise the soundness theorem for our reset system $\cmuPAord$  (\autoref{thm:Graham-soundness}),  and then we appeal to conservativity results in~\autoref{thm:moellerfeld-psca-cons-over-muPA}. 

We first recall fragments of  ``second-order" arithmetic and some reverse mathematics of fixed point theorems developed in~\cite{Curzi023}. Then, we proceed with arithmetising the soundness argument.

\subsection{Language  and theories of second-order arithmetic}
Let us recall
$\langsoarith$, the language of second-order arithmetic, e.g.\ as given in \cite{Simpson99:monograph}.
It extends the language of arithmetic $\langarith$ by:
\begin{itemize}
    \item an additional sort of \emph{sets}, whose variables are written $X,Y $ etc. Individuals of $\langarith$ are considered of \emph{number} sort.
    \item an \emph{elementhood} (or \emph{application}) relation $\in$ relating the number sort to the set sort. I.e.\ there are (atomic) formulas $t \in X$ when $t $ is a number term and $X$ is a set variable. 
    % (We will not consider any non-variable set terms here.)
\end{itemize}

When speaking about the free variables of a formula, we mean both set variables and individual variables.

We shall assume a De Morgan basis of connectives, namely $\lor, \land, \exists, \forall$ with negation only on atomic formulas.
Hence, we say that a formula $\phi$ is \emph{positive} (or \emph{negative}) in $X$ if no (every, resp.) subformula $t \in X$ occurs under a negation.

We will use the term $\Sij 1 k$ to denote formulas provably equivalent under $\PSCA$ to a formula of the form $\exists X_0 \forall X_1 \hdots Q_n \phi$ with $\phi$ arithmetical. Similarly for $\Pij 1 k $ and $\Dij 1 k$. 

We use the notation $\langle - \rangle$ for the code of a tuple. If the tuple contains a set then $\langle - \rangle$ will be a second-order object. 

For any set $F$ we can consider $F$ as a function $F: \Nat \rightarrow SET$ by setting $F(x):= \{ y \ | \ \langle x , y \rangle \in F \}$. If $C$ is a collection of sets defined by a formula $\phi(X)$ then we write $F: A \rightarrow C$ for a $\Dij 1 1 ( \phi )$ formula expressing that, for all $a \in A$, we have $F(a) \in C$. In general, for any formula $\psi(Y,x)$ we have that $\psi(F(x),x)$ is $\Dij 1 1 (\psi)$, and so we can use $F(x)$ freely without worrying about increasing second-order quantifier complexity. 

We shall work with common subsystems of second-order arithmetic, as found in textbooks such as \cite{Simpson99:monograph}, and assume basic facts about them.

In particular, recall that $\ACA$ is a two-sorted extension of basic arithmetic by:

\begin{itemize}\setlength\itemsep{0.1cm}
    \item \emph{Arithmetical comprehension.} $\exists X \forall x (X(x) \liff \phi(x))$ for each arithmetical formula $\phi(x)$.
    \item \emph{Set induction.} $\forall X (X(0) \limp \forall x (X(x) \limp X(\succ x)) \limp \forall x X(x)) $
\end{itemize}

From here, 
$\PSCA$ is the extension of $\ACA$ by the comprehension schema for all $\Pij 1 2 $ formulas.
We shall freely use established principles of $\PSCA$ (as seen, e.g., in \cite{Simpson99:monograph}) such as $\Sij 1 2 $-comprehension, $\Sij 1 2$-choice and $\Sij 1 2 $-dependent choice

\begin{equation}\label{eqn:axiom-of-choice}
    (\text{$\Sij 1 2 $-choice}) \ \ \ \forall x \exists Y \phi(x,Y) \leftrightarrow \exists F: \Nat \rightarrow SET \ \forall x \phi(x , F(x)) \ \ \ \phi \in \Sij 1 2
\end{equation} 

\begin{equation}\label{eqn:dependent-choice}
    (\text{$\Sij 1 2 $-dependent choice}) \ \ \ \forall X \exists Y \phi(X,Y) \implies \exists F: \Nat \rightarrow SET \ \forall n \phi(F(n) , F(n+1)) \ \ \ \phi \in \Sij 1 2
\end{equation}

 %   Assume $\psi = \psi_0 \land \psi_1$. Since $Z$ only occurs once in $\psi$ we can assume wlog that $Z$ occurs once in $\psi_0(Z)$ and not in $\psi_1$. By inductive hypothesis 

 %   $$\PSCA \proves Q_0 X_0 \psi_0 ( Q_1 X_1 \hdots Q_k X_k \phi' ) \land \psi_1 \leftrightarrow Q_0 X_0 (\psi_0 ( Q_1 X_1 \hdots Q_k X_k \phi' ) \land \psi_1)$$.

 %   Assume $\psi = \forall y \psi_0$ and $Q_0= \exists$. By inductive hypothesis 
    
 %   The three remaining first-order quantifier cases are trivial. 

For the formalised soundness argument we will need that $\Dij 1 2$ formulas are closed under positive arithmetical combinations (over $\PSCA$).

\begin{thm} \label{thm:arithm-comb-of-delta12}

Let $\psi(\phi_0, \hdots \phi_n)$ be a positive arithmetical combination of $\phi_0, \hdots \phi_n$. If $\phi_0, \hdots \phi_n \in \Sij 1 2$ (resp., $\Pij 1 2$) then $\psi(\phi_0, \hdots \phi_n) \in \Sij 1 2$ (resp., $\Pij 1 2$).

\end{thm}

\begin{proof}
We will just prove the result for $\Sij 1 2$, the proof for $\Pij 1 2$ is very similar.

We proceed by induction on the structure of $\psi$. 

The cases where $\psi$ is atomic or where $\psi = \phi_i$ for some $i \leq n$ are obvious.

Assume $\psi(\vec \phi) = \psi_0(\vec \phi) \land \psi_1(\vec \phi)$. By inductive hypothesis we can write $\psi_0(\vec \phi) = \exists X_0 \forall X_1 \psi'_0$ and $\psi_1(\vec \phi) = \exists Y_0 \forall Y_1 \psi'_1$ with $\psi_0'$ $ \psi_1'$ arithmetical. Then by standard prenexing we get 

$$\PSCA \proves \exists X_0 \forall X_1 \psi'_0 \land  \exists Y_0 \forall Y_1 \psi'_1 \liff \exists X_0 \exists Y_0 \forall X_1 \forall Y_1 ( \psi'_0 \land    \psi'_1)$$

The $\lor$ case is identical.

Assume $\psi(\vec \phi) = \forall x \psi_0(\vec \phi)$. By inductive hypothesis we have $\psi_0(\vec \phi) = \exists X_0 \forall X_1 \psi'_0$ with $\psi'_0$ arithmetical. Using \autoref{eqn:axiom-of-choice} for the first equivalence and pure logic for the second we get: 

$$\PSCA \proves \forall x \exists X_0 \forall X_1 \psi'_0(X_0) \liff \exists F \forall x \forall X_1 \psi'_0(F(x))  \liff \exists F \forall X_1 \forall x  \psi'_0(F(x))$$

with $\psi'_0(F(x) )\in \Dij 1 1$. 

The case of $\exists$ is similar. 
\end{proof}

\subsection{Countable orders in $\PSCA$} \label{subsection:formalized-ordinals}

A basic theory of (countable) ordinals can be developed  even within weak second-order theories, as shown in [Sim99] and also comprehensively surveyed in [Hir05]. In this paper, however, our primary focus is on $\PSCA$. Therefore, we will omit specifying the weakest theories in which each result can be formalised.

A (countable) binary relation is a pair $(X, \leq)$ where $X$ and $\leq$ are sets, the latter construed as ranging over pairs. We say that $(X, \leq)$ is a partial order, written $PO(X, \leq)$, if:

\begin{itemize}
    \item $\forall x \in X x \leq x$
    \item $\forall x, y \in X(x \leq y \leq x \rightarrow x = x)$
    \item $\forall x, y, z(x \leq y \leq z \rightarrow x \leq z)$
\end{itemize}

$(X, \leq)$ is a total order, written $TO(X, \leq)$, if it is a partial order that is total, that is,
\begin{itemize}
    \item $ \forall x, y(x \leq y \lor y \leq x)$
\end{itemize}

Given a relation $\leq$, we may write $<$ for its strict version, given by
$x < y := x \leq y \land \neg y \leq x$.

We employ similar notational conventions for other related order-theoretic binary symbols. 
We say that a binary relation $(X, \leq)$ is well-founded, written $WF(X, \leq)$, if:
\begin{itemize}
    \item $\forall f : \Nat \rightarrow X \exists x \neg f (x + 1) < f (x)$
\end{itemize}

$(X, \leq)$ is a well-order, written $WO(X, \leq)$, if it is a well-founded total order, i.e.:

$WO(X, \leq) := TO(X, \leq) \land WF(X, \leq)$

Henceforth we shall write $\alpha$, $\beta$ etc. to range over countable binary relations. If $\alpha = (X, \leq)$ we may write $x \leq_{\alpha} y := x, y \in X \land x \leq y$, and similarly $x <_{\alpha} y$ for $x, y \in X \land x < y$. We may also write simply $x \in  \alpha$ instead of $x \in  X$, as abuse of notation.

It is not hard to see that $\mathsf{ACA}_0$ admits an induction principle over any provable well-order (see, e.g., [Sim99]):
\begin{equation}
    \label{eq:transf-ind}
    \ \mathsf{ACA}_0 \proves 
    \quad WO(\alpha) \rightarrow  \
    \forall X(\forall x \in  \alpha(\forall y <_{\alpha} x  
    X(y) \rightarrow X(x)) \rightarrow \forall x \in  \alpha X(x))
\end{equation}

And so, in particular, in $\PSCA$ we have transfinite induction on arithmetical combinations of $\Pij 1 2 $-formulas (using $\Pij 1 2$-comprehension). 

Following Simpson in [Sim99], given $\alpha$, $\beta \in  WO $ we write $\alpha \prec  \beta$ if there is an order-isomorphism from $ \alpha$ onto a proper initial segment of $\beta$. We also write $\alpha \approx   \beta$ if $\alpha$ and $\beta$ are order-isomorphic, and $\alpha \preceq    \beta$ if $\alpha \prec  \beta \lor \alpha \approx   \beta$. Thanks to the uniqueness of comparison maps, we crucially have

\begin{prop}
$\prec$, $ \preceq$ and $\approx$ are all provably $\Dij 1 1$ in $\PSCA$. 
\end{prop} 

We shall now state a number of well-known facts about comparison, all of which may be found in, e.g., [Sim99] or [Hir05].

\begin{prop} \label{prop:props-ordinal-comparison}
Let $\alpha$, $\beta$, $\gamma$ be well-orders. We have the following in $\PSCA$:

\begin{enumerate}
    \item (Comparison is a preorder)
    \begin{enumerate}
        \item $\alpha \preceq    \alpha$
        \item $\alpha \preceq    \beta \preceq    \gamma     \rightarrow \alpha \preceq    \gamma$
    \end{enumerate}
    \item (Comparison is pseudo-antisymmetric) If $\alpha \preceq    \beta \preceq    \alpha$ then $\alpha \approx   \beta$.
    \item (Comparison is total) $\alpha \preceq    \beta \lor \beta \preceq    \alpha$
    \item (Comparison is well-founded) 
    
    $\forall F : \Nat \rightarrow WO \ (( \forall x \ \exists y > x  \ F (y + 1) \prec  F (y) ) \rightarrow \exists x  ( F (x) \prec  F (x+1) ))$   
\end{enumerate}

\end{prop} 
 
We shall assume basic ordinal existence principles, in particular constructions for successor ($\succ \alpha$), addition ($\alpha + \beta$) and maximum ($\max(\alpha, \beta)$), initial segments ($\alpha_b$ for $b \in  \alpha$), all definable and satisfying characteristic properties provably in $\PSCA$.

\subsection{Knaster-Tarski Theorem and approximants} \label{subsection:formalized-knaster-tarski}

We have two ways of interpreting the $\mu$- and $\nu$-binders in second-order arithmetic. The first we might call the ``impredicative way'' given by:

\begin{equation}
    \label{eq:mu-in-SO}
    t \in \mu  X x . \phi 
    \quad := \quad
    \forall X (\forall x (\phi(X,x) \limp x \in X) \limp t \in X)
\end{equation} 
\begin{equation}
    \label{eq:nu-in-SO}
    t \in \nu  X x .\phi 
    \quad := \quad
    \exists X (\forall x (x \in X \rightarrow \phi(X,x) ) \land t \in X)
\end{equation}

Another way of interpreting $\mu$ and $\nu$ is through approximants:

\begin{defi}[Approximants]
We write  $t \in \mu^{\alpha} X x. \phi   $ for the following formula: 
\[   \exists F: \alpha \rightarrow SET \ \exists a \in \alpha \ (\forall b \in \alpha \ \forall y \ (y \in F(b)  \rightarrow \exists c <_\alpha b \ \phi (F(c) , y)) \land  t \in F(a) )
\]
and  $t \in \mu^{\WO} X x. \phi $ for $   \exists \alpha \in WO \ t \in \mu^{\alpha} X x. \phi$. Dually, we write $t \in \nu^{\alpha} X x. \phi $ for the formula
\[
  \forall F : \alpha \rightarrow SET  \ \forall a \in \alpha \ (\forall b \in \alpha \ \forall y \ ( \forall c <_\alpha b \ \phi (F(c) , y) \rightarrow  y \in F(b)  ) \rightarrow  t \in F(a) )
\]
and $t \in \nu^{\WO} X x. \phi $ for $   \forall \alpha \in WO \ t \in \nu^{\alpha} X x. \phi$.

% \begin{tabular}{lll}
%      $t \in \mu^{\alpha} X x. \phi   $ & $:=$ & $   \exists F: \alpha \rightarrow SET \ \exists a \in \alpha \ (\forall b \in \alpha \ \forall y \ (y \in F(b)  \rightarrow \exists c <_\alpha b \ \phi (F(c) , y)) \land  t \in F(a) )$ \\
     
%     $t \in \mu^{\WO} X x. \phi $ & $:=$ & $   \exists \alpha \in WO \ t \in \mu^{\alpha} X x. \phi$
% \\ \\ 
%      $t \in \nu^{\alpha} X x. \phi   $ & $:=$ & $   \forall F : \alpha \rightarrow SET  \ \forall a \in \alpha \ (\forall b \in \alpha \ \forall y \ ( \forall c <_\alpha b \ \phi (F(c) , y) \rightarrow  y \in F(b)  ) \rightarrow  t \in F(a) )$ \\
      
%     $t \in \nu^{\WO} X x. \phi $ & $:=$ & $   \forall \alpha \in WO \ t \in \nu^{\alpha} X x. \phi$
% \end{tabular}

\end{defi}

 For the remainder of this section, assume $\phi(X,x) \in \Dij 1 2$  and assume $\phi(X,x)$ is positive in $X$.  
 % \lukas{Lukas: check with Anupam if you need  polarity stuff from Anupam's paper.}

\begin{lem}[Dualities] \label{lemma:pi12-proves-mu-nu-duality}
    $\PSCA$ proves 

    \begin{tabular}{lll}
        $t \in \mu  X x . \phi(X,x)$ & $\leftrightarrow$ & $t \notin \nu  X x . \neg \phi(\neg X,x) $\\
        $t \in \mu^{\alpha} X x . \phi(X,x)$ & $\leftrightarrow$ & $   t \notin \nu^{\alpha} X x . \neg \phi(\neg X,x) $ \\
       $t \in \mu^{WO} X x . \phi(X,x)$ & $\leftrightarrow$ & $   t \notin \nu^{WO} X x . \neg \phi(\neg X,x) $
    \end{tabular}

\end{lem}

\begin{proof}
    This is done merely by repeating the standard textbooks proofs within $\PSCA$ using pure logic. Note that no comprehension is needed. 
\end{proof}

The last results of this section are all proven in \cite{Curzi023} for $\mu$-terms (the dual properties for $\nu$-terms can  be obtained using \autoref{lemma:pi12-proves-mu-nu-duality}).

The first result shows that approximants are $\Dij 1 2$. This is shown by utilising the uniqueness of $F$ (under certain constraints). 

\begin{prop} \label{prop:delta12-characterization-of-approximants}

If $\phi(X,x) \in \Dij 1 2$ and $\PSCA \proves WO(\alpha)$ then $\mu^{\alpha} X x . \phi$ and $\nu^{\alpha} X x . \phi$ are in  $\Dij 1 2$

\end{prop}

The next result shows that $\PSCA$ has access to the recursive characterizations of $\mu$- and $\nu$-approximants respectively, i.e. 
\[
\begin{array}{c}
    \mu^{\alpha}Xx.\phi = \displaystyle \bigcup_{\beta < \alpha} \phi ( \mu^{\beta}Xx . \phi )\qquad \text{and} \qquad
    \nu^{\alpha}Xx.\phi  = \displaystyle \bigcap_{\beta < \alpha} \phi ( \nu^{\beta}Xx . \phi )
\end{array} 
\]

\begin{prop} [Recursion] \label{prop:formalized-approximant-recursion}

$\PSCA$ proves

\begin{tabular}{lll} 
    $t \in \mu^{\alpha} X x . \phi $ & $\leftrightarrow$ & $\exists \beta \prec \alpha \   \phi(\mu^{\beta} X x . \phi , t )$ \\
   $t \in \nu^{\alpha} X x . \phi $ & $\leftrightarrow$ & $\forall \beta \prec \alpha \   \phi ( \nu^{\beta} X x . \phi , t )$ 
\end{tabular}

\end{prop}

Next, $\PSCA$ needs to ``know'' that $\nu X x .\phi$ and $\mu X x .\phi$ are fixed points.

\begin{prop} [Knaster-Tarski] \label{prop:formalized-knaster-tarski}

$\PSCA$ proves 

\begin{tabular}{l}

  $\forall y \ ( y \in \mu X x . \phi(X,x) \leftrightarrow \phi (\mu X x . \phi,y  )$ \\
  $\forall y \ ( y \in \nu X x . \phi(X,x) \leftrightarrow \phi (\nu X x . \phi, y )$
\end{tabular}

\end{prop}

Finally, we can show that the two interpretations of $\mu$ and $\nu$ are provably equivalent in $\PSCA$. This result uses \autoref{prop:delta12-characterization-of-approximants}, \autoref{prop:formalized-approximant-recursion} and \autoref{prop:formalized-knaster-tarski}, as well as some additional preliminary results from \cite{Curzi023} that are not mentioned here. 

\begin{thm}[Fixed points characterisation]\label{thm:formalized-dual-characterization}

$\PSCA$ proves 

\begin{tabular}{l} 
  $\forall y ( y \in \mu  X x . \phi \leftrightarrow y \in \mu^{WO} X x . \phi ) $ \\
  $\forall y ( y \in \nu  X x . \phi \leftrightarrow y \in \nu^{WO} X x . \phi ) $
\end{tabular}

\end{thm}

By standard prenexing and $\Sij 1 2$-choice we see that $\mu X x . \phi \in \Pij 1 2$, and using standard prenexing and \autoref{prop:delta12-characterization-of-approximants} we see that $\mu^{WO} Xx.\phi \in \Sij 1 2$. Thus,  from \autoref{thm:formalized-dual-characterization} we have that both interpretations of $\mu$ are $\Dij 1 2$. The same goes for $\nu$. 

\begin{cor}
    
 \label{cor:delta12-characterization-of-mu-nu}
If $\phi(X,x) \in \Dij 1 2$ then $\mu  X x . \phi$, $\mu^{WO} X x . \phi$, $\nu  X x . \phi$ and $\nu^{WO} X x . \phi$ are all $\Dij 1 2$. 
\end{cor}

\subsection{Interpreting $\cmuPA$ and $\cmuPAord$ in $\PSCA$}

\label{sec:subsyss-so-arith}
\label{sec:ids-in-pca}

We may define an interpretation $(\_)^\dagger$ from $\cmuPA$ to $\PSCA$ defined by $(t \in \mu X x \phi)^\dagger\dfn t \in \mu  X x \phi$ and $(t \in \nu X x \phi)^\dagger\dfn t \in \nu  X x \phi$.

\begin{thm}\label{thm:image-of-muPA-translation-Delta12}
If $\phi$ is a formula of $\cmuPA$ then $\phi^\dagger$ is $\Dij 1 2$.
\end{thm} 

\begin{proof}

This follows from \autoref{cor:delta12-characterization-of-mu-nu} and \autoref{thm:arithm-comb-of-delta12}. 
\end{proof}

\autoref{thm:image-of-muPA-translation-Delta12} allows $\PSCA$ to readily verify the axioms of $\muPA$ using $\Pij 1 2$-comprehension, and so $\muPA$ can be considered a fragment of $\PSCA$. In fact, by a result of M\"ollerfeld's, $\PSCA$ not only extends $\muPA$ but does so conservatively.

\begin{thm}
    [Implied by \cite{moellerfeld02:gen-ind-dfns}]
    \label{thm:moellerfeld-psca-cons-over-muPA}
    $\PSCA$ is arithmetically conservative over $\muPA$.
\end{thm}

This is a nontrivial result in proof theory whose details we shall not recount.
We will use this theorem as a ``black box'' henceforth.

Similarly, we can define an interpretation $(\_)^\dagger$ from $\cmuPAord$ to $\PSCA$ defined by $(t \in \mu X x. \phi)^\dagger \dfn t \in \mu  X x .\phi$, $(t \in \nu X x .\phi)^\dagger\dfn t \in \nu^{WO}  X x .\phi$ and $(t \in \nu^{\kappa} .X x \phi )^\dagger \dfn t \in \nu^{\kappa}  X x .\phi$.
Again, we have:

\begin{thm}
    
\ \label{thm:image-of-muPAord-translation-Delta12}
If $\phi$ is a formula of $\cmuPAord$ then $\phi^\dagger$ is $\Dij 1 2$.

\end{thm}

\begin{proof}
    This follows from \autoref{prop:delta12-characterization-of-approximants}, \autoref{cor:delta12-characterization-of-mu-nu} and \autoref{thm:arithm-comb-of-delta12}. 
\end{proof}

\subsection{Formalizing $\muPAord$ satisfaction in $\PSCA$}

In what follows we will employ standard metamathematical notations and conventions for coding, e.g.\ we write $\code E$ for the G\"odel code of an expression $E$.
Also, when it is not ambiguous, we shall typically use the same notation for meta-level objects/operations and their object-level (manipulations on) codes, as a convenient abuse of notation.
As usual, a predicate language cannot contain it's own universal truth predicate for Tarskian reasons.
However we may define \emph{partial} truth predicates for fragments of the language, in particular, each of our partial satisfaction predicates will be relativized to a single $\cmuPAord$ proof. For the rest of this section, fix such a proof $\der$. 
Since only finitely many first-order variables, ordinal variables and set-terms occur in $\der$, we can assume that the first order variables, ordinal variables and set-terms occurring in $\der$ are among, respectively, $x_i$, $\kappa_j$, $\mu X_i x_i . \phi_i$, $\nu X_i x_i . \phi_i$, $\nu^{\kappa_j} X_i x_i . \phi_i$ and $X_i$ for $i,j$ less than some $k$. 

\begin{defi} \label{def:formalized-sat-pred}
    Let $\chi(z,Z, \vec A, \vec \alpha )$ be an $\langsoarith$ formula expressing the statement in \autoref{fig:formalized-sat-pred}. Since $Z$ occur positively in $\chi(z,Z, \vec A, \vec \alpha )$ we can define

$$ \sat{\der}(\rho,m , \vec A , \vec \alpha) := \quad \langle \rho,m \rangle \in \mu Z z. \chi(z,Z, \vec A, \vec \alpha )  $$
    
Note that by \autoref{thm:image-of-muPAord-translation-Delta12} we have that $\sat{\der}(\rho,m , \vec A , \vec \alpha ) \in \Dij 1 2$.
\end{defi}

Note that all formulas in $\der$ are assumed to be closed, and so all set-variables occur positively in any subformula of a formula in $\der$. For this reason there is no need  for a clause for $t \not \in X_i$ in \autoref{fig:formalized-sat-pred}.

Also note that for our formalized satisfaction predicate we are not using extended assignments $\rho$ that associate ordinals to ordinal variables. Instead, $\rho$ is a finite object (coded by a number) that just assigns numbers to first-order terms, while the ordinals assigned to ordinal variables are second-order parameters $\vec \alpha$ of the formula. Similarly, $\rho$ does not assign sets to set-terms, as this is treated separately with the second-order parameters $\vec A$. 

There are two main properties that we need our partial satisfaction predicate to satisfy. The first is reflection. 

\begin{thm}
    [Reflection]
    \label{reflection}
    
Let $\phi(\vec X, \vec x, \vec \kappa)$ be a subformula of a formula in $\der$. Then $\PSCA$ proves $\sat{\der}(\rho, \code{\phi(\vec X,\vec x,  \vec \kappa)}, \vec A,  \vec \alpha) \liff \phi(\vec A, \rho(\vec x), \vec \alpha)$.

\end{thm}

\begin{proof}
    This is proved by a meta level induction on the structure of $\phi$ using \autoref{fig:formalized-sat-pred}.
\end{proof}

\begin{figure}
    \centering
     \[ 
    \begin{array}{rl}
        & z=\langle \rho , m \rangle \quad \land \\
      (  &  m = \code{s=t} \land \rho(s) = \rho(t) \\
    \lor & m = \code{s \neq t} \land \rho(s) \neq \rho(t) \\
    \lor & m = \code{s<t} \land \rho ( s) <\rho(t)  \\
    \lor & m = \code{\neg (s<t) } \land \neg ( \rho ( s) < \rho(t) )  \\    
    \lor & m = \code{\phi \lor \psi} \land \left( \langle \rho,\code \phi \rangle \in Z \lor \langle \rho,\code \psi \rangle \in Z \right) \\
    \lor & m = \code{\phi \land \psi} \land \left( \langle \rho,\code \phi \rangle \in Z \land \langle \rho,\code \psi \rangle \in Z \right) \\
    \lor & m = \code{\exists x \phi} \land \exists n \   \langle \rho \{x\mapsto n\}, \code \phi \rangle \in Z \\
    \lor & m = \code{\forall x \phi} \land \forall n \ \langle \rho \{x\mapsto n\}, \code \phi \rangle \in Z \\
    \noalign{\smallskip}
    \lor & \bigvee\limits_{i=0}^k (m = \code{t \in \mu X_i x_i . \phi_i} \land \rho(t) \in  \mu Y_i y_i . \phi_i( \rho(x_0), \hdots y_i , \hdots \rho (x_k), A_0 , \hdots Y_i , \hdots A_k, \vec \alpha) )  \\
    \lor & \bigvee\limits_{i=0}^k (m = \code{t \in \nu X_i x_i . \phi_i} \land \rho(t) \in  \nu^{WO} Y_i y_i .\phi_i( \rho(x_0), \hdots y_i , \hdots \rho (x_k), A_0 , \hdots Y_i , \hdots A_k, \vec \alpha) ) \\
    \lor & \bigvee\limits_{i=0}^k \bigvee\limits_{j=0}^{k} (m = \code{t \in \nu^{\kappa_j} X_i x_i . \phi_i} \land \rho(t) \in  \nu^{\alpha_j} Y_i y_i . \phi_i( \rho(x_0), \hdots y_i , \hdots \rho (x_k), A_0 , \hdots Y_i , \hdots A_k, \vec \alpha) ) \\
    \lor & \bigvee\limits_{i=0}^k (m = \code{t \in X_i} \land \rho(t) \in A_i ) \quad ) \\
    \end{array}
    \]
    \caption{Construction of partial satisfaction predicate.}
    \label{fig:formalized-sat-pred}
\end{figure}

\begin{figure}
    \centering
     \[
    \begin{array}{llrl}
    
      & \forall  \vec \alpha \in WO \ \forall \rho, s, t , \phi, \psi, x , \vec A &  & \\

      & & & \\
        
  &   \sat{\der}(\rho, \code{s=t}  , \vec A, \vec \alpha) & \leftrightarrow   &  \rho(s) = \rho(t) \\
     
  &  \sat{\der}(\rho, \code{s \neq t} , \vec A, \vec \alpha) &  \leftrightarrow & \rho(s) \neq \rho(t) \\
    
  &  \sat{\der}(\rho,\code{s<t}, \vec A, \vec \alpha) &      \leftrightarrow & \rho ( s) <\rho(t)  \\
    
  &  \sat{\der}(\rho,\code{\neg (s<t) }, \vec A, \vec \alpha) &  \leftrightarrow & \neg ( \rho ( s) < \rho(t) )  \\ 
    
  &  \sat{\der}(\rho, \code{\phi \lor \psi} , \vec A, \vec \alpha) &   \leftrightarrow &  \sat{\der} (\rho,\code \phi, \vec A , \vec \alpha) \lor \sat{\der} (\rho,\code \psi , \vec A , \vec \alpha ) \\
    
  &  \sat{\der}(\rho, \code{\phi \land \psi} , \vec A, \vec \alpha) &  \leftrightarrow & \sat{\der} (\rho,\code \phi, \vec A , \vec \alpha) \land \sat{\der} (\rho,\code \psi, \vec A ,\vec \alpha) \\
    
   & \sat{\der}(\rho, \code{\exists x \phi} , \vec A, \vec \alpha) &   \leftrightarrow & \exists n\, \sat{\der} (\rho \{x\mapsto n\}, \code \phi, \vec A , \vec \alpha)\\
    
  &  \sat{\der}(\rho, \code{\forall x \phi} , \vec A, \vec \alpha) &  \leftrightarrow & \forall n\, \sat{\der} (\rho \{x\mapsto n\}, \code \phi, \vec A , \vec \alpha ) \\
    
    \noalign{\smallskip}
    
\bigwedge \limits_{i=0}^{k} &     \sat{\der}(\rho, \code{t \in \mu X_i x_i . \phi_i} , \vec A, \vec \alpha) &   \leftrightarrow & \rho(t) \in  \mu Y_i y_i . \sat{\der} (\rho\{ x_i \mapsto y_i \}, \code {\phi_i} , A_0 , \hdots Y_i , \hdots A_k , \vec \alpha ) ) \\
    
\bigwedge \limits_{i=0}^{k}  &  \sat{\der}(\rho, \code{t \in \nu X_i x_i . \phi_i} , \vec A, \vec \alpha) &   \leftrightarrow & \rho(t) \in  \nu^{WO} Y_i y_i . \sat{\der} (\rho\{ x_i \mapsto y_i \}, \code {\phi_i}, A_0 , \hdots Y_i , \hdots A_k , \vec \alpha ) ) \\
    
\bigwedge \limits_{i,j=0}^{k} &   \sat{\der}(\rho, \code{t \in \nu^{\kappa_j} X_i x_i . \phi_i} , \vec A, \vec \alpha) &   \leftrightarrow &  \rho(t) \in  \nu^{\alpha_j} Y_i y_i . \sat{\der} (\rho\{ x_i \mapsto y_i \}, \code {\phi_i} , A_0 , \hdots Y_i , \hdots A_k , \vec \alpha) ) \\
    
\bigwedge \limits_{i=0}^{k} &    \sat{\der}(\rho, \code{t \in X_i} , \vec A, \vec \alpha) &   \leftrightarrow &  \rho(t) \in A_i  \\
    \end{array}
    \]
    \caption{Inductive characterisation of the satisfaction predicate.}
    \label{fig:sat-characteristic-props}
\end{figure}

   % \lor & m = \code{\phi \lor \psi} \land \left(\sat{\der} (\rho,\code \phi, \vec A , \vec \alpha) \lor \sat{\der} (\rho,\code \psi , \vec A , \vec \alpha )\right) \\
 %   \lor & m = \code{\forall x \phi} \land \forall n\, \sat{\der} (\rho \{x\mapsto n\}, \code \phi, \vec A , \vec \alpha ) \\
 %   \lor & \bigvee\limits_{i=1}^k (m = \code{t \in \mu X_i x_i . \phi_i} \land \rho(t) \in  \mu Y_i y_i . \sat{\der} (\rho\{ x_i \mapsto y_i \}, \code {\phi_i} , A_0 , \hdots Y_i , \hdots A_k , \vec \alpha ) ) \\
 %   \lor & \bigvee\limits_{i=1}^k (m = \code{X_i(t)} \land \rho(t) \in A_i ) \quad ) \\
 %   \caption{Inductive characterisation of the satisfaction predicate.}
 %   \label{fig:sat-characteristic-props}

 %   \sat{\der}(\rho,m,\vec A) \quad \liff \quad 
  %  \begin{array}{rl}

 %   [Formalised relative satisfaction]

The second property we need is that our partial satisfaction predicate satisfies the proper inductive characterization. 

\begin{thm}
    [Formalised satisfaction]
    \label{theorem:formalised-satisfaction}

 $\PSCA$ proves the inductive characterization of $\sat{\der}(\rho,m, \vec A , \vec \alpha)$ given in \autoref{fig:sat-characteristic-props}.
\end{thm}

\begin{proof}
For the atomic cases, $\land$, $\lor$, $\exists$ and $\forall$ this follows straightforwardly from \autoref{fig:formalized-sat-pred} and \autoref{prop:formalized-knaster-tarski}. For the remaining cases we need \autoref{reflection}. 

Concerning the case of $\mu$-formulas, using \autoref{def:formalized-sat-pred} and \autoref{prop:formalized-knaster-tarski} it is clear that $\PSCA$ proves

$$\sat{\der}(\rho, \code{t \in \mu X_i x_i . \phi_i} , \vec A, \vec \alpha) \leftrightarrow \rho(t) \in  \mu Y_i y_i . \phi_i( \rho(x_0), \hdots y_i , \hdots \rho (x_k), A_0 , \hdots Y_i , \hdots A_k, \vec \alpha) $$

By \autoref{reflection} we have that $\PSCA$ proves

$$ \phi_i( \rho(x_0), \hdots y_i , \hdots \rho (x_k), A_0 , \hdots Y_i , \hdots A_k, \vec \alpha) \leftrightarrow  \sat{\der} (\rho\{ x_i \mapsto y_i \}, \code {\phi_i} , A_0 , \hdots Y_i , \hdots A_k , \vec \alpha ) $$

and so $\PSCA$ proves

\[
\def\arraystretch{1.2}
\begin{array}{c}
 \rho(t) \in  \mu Y_i y_i . \phi_i( \rho(x_0), \hdots y_i , \hdots \rho (x_k), A_0 , \hdots Y_i , \hdots A_k, \vec \alpha) \\
     \leftrightarrow\\
     \rho(t) \in  \mu Y_i y_i . \sat{\der} (\rho\{ x_i \mapsto y_i \}, \code {\phi_i} , A_0 , \hdots Y_i , \hdots A_k , \vec \alpha )
      
\end{array}    
\]

from which we can conclude. 

The remaining cases are proved similarly. 
\end{proof}

Recall that we write $\neg \phi$ for the De Morgan dual of $\phi$. 

\begin{prop}

    \label{prop:formalized-sat-respects-negation}
$\PSCA$ proves 
\[\forall  \vec \alpha \in WO \ \forall \rho, \phi , \vec A \ \neg  \sat{\der}(\rho, \code{\phi} , \vec A , \vec \alpha) \liff \sat{\der}(\rho, \code{ \neg \phi}, \vec A , \vec \alpha)\]

where $\phi$ is a subformula of a formula in $\der$. 

\end{prop}

\begin{proof}
    This is the usual proof by induction on the structure of $\phi$ carried out in $\PSCA$ using \autoref{theorem:formalised-satisfaction}.
\end{proof}

For the formalized soundness proof we also need to extend our satisfaction predicate $\sat{\der}$ so that it works with constraints and extended sequents. This is done in the expected way.

\subsection{Formalizing soundness for $\cmuPAord$}

In the last part of this section we show that $\PSCA$ can formalise soundness for $\cmuPAord$, i.e., the proof of \autoref{thm:Graham-soundness}. 

Recall that in \autoref{thm:Graham-soundness}, after showing the existence of the  branch witnessing falsity, we invoke the reset condition to extract a stable reset variable along the false branch and thereby get a contradiction. To mimic this step inside $\PSCA$ it is crucial that $\PSCA$ ``knows'' that the proof satisfies the reset condition. 

\begin{lem} \label{lemma:reset-condition-provable}
Let $\der$ be a $\cmuPAord$ proof. Then $\PSCA$ proves that for all branches $\branch$ in $\der$ there exists a stable reset variable $\kappa$ along $\branch$.    
\end{lem}

\begin{proof}

     The idea is standard and is covered in detail in \cite{Das20:ca-log-comp}. 

     It is provable in $\PSCA$ (and much weaker theories) that for a $\cmuPAord$ derivation $\der$, satisfying the reset condition is equivalent to language inclusion of two Büchi automata which can both be computed from $\der$. Since by \cite{KMPS19:buchi-rev-math} language inclusion for Büchi automata is decidable it follows that, if $\der$ satisfies the reset condition, then $\PSCA$ proves that fact. 
\end{proof}

\begin{thm}[Formalising soundness]\label{thm:formalisation-of-soundness}
   Let $\der$ be a $\cmuPAord$ proof with conclusion $\constr_{\der} : \Gamma_{\der}$. Then  \[\PSCA\vdash \forall \vec \alpha \in WO \ \forall \rho, \vec A  \ \sat{\der} (\rho , \code{ \constr_{\der} : \Gamma_{\der} } , \vec A, \vec \alpha)\]
\end{thm}

\begin{proof}

Note that, unless otherwise indicated, all reasoning is assumed to happen within $\PSCA$.

Assume for contradiction that $\neg \sat{\der} (\rho_{\der} , \code{\constr_{\der} : \Gamma_{\der}} , \vec A, \vec \alpha_{\der} ) $. As in \autoref{thm:Graham-soundness} we will show the existence of a false branch satisfying certain conditions. The key to this is \autoref{lemma:one-step-branch-lemma} which says that whenever we have a false extended sequent then we can find a premise which is false and satisfies the right conditions on ordinal variables. To state \autoref{lemma:one-step-branch-lemma} inside $\PSCA$ we define the following formulas:

\[
 \xi(X)\dfn 
\begin{pmatrix}
X = \langle \rho  , \code{\constr : \Gamma} , \vec \alpha \rangle \ \land \vec \alpha \in WO  \\
      \land \\
     \constr : \Gamma \text{ is an extended sequent in }\der \\
     \land \\
   \neg \sat{\der} (\rho, \code{\constr : \Gamma} , \vec A , \vec \alpha )
\end{pmatrix}
\]

 and 
 
\[
\chi(X,Y)\dfn 
\begin{pmatrix}
    X = \langle  \rho ,  \code{\constr : \Gamma} , \vec \alpha \rangle \ \land \vec \alpha \in WO \  \land Y = \langle \rho' , \code{\constr' : \Gamma'} , \vec \alpha' \rangle \ \land \vec \alpha' \in WO \  \\
    \land\\
    \constr' : \Gamma' \text{ is a premise of  }\constr : \Gamma \text{ in }\der\\
    \land \\
    \forall n \leq k \left( \kappa_n \in \constr \cap \constr' \rightarrow 
    ( ( \alpha_n' \preceq \alpha_n)  \land \text{(the rule is }\rs {\kappa_n} \rightarrow \alpha_n' \prec \alpha_n  ) )\right)
\end{pmatrix}
\]

Note that both $\xi(X)$ and $\chi(X,Y)$ are arithmetical in $\sat{\der}$, $\prec$ and $WO$, and so they are $\Dij 1 2$ by \autoref{thm:arithm-comb-of-delta12}.

Let $Z_{\der} = \langle \rho_{\der }, \code{  \constr_{\der} : \Gamma_{\der} } , \vec \alpha_{\der}  \rangle$. The formal equivalent of \autoref{lemma:one-step-branch-lemma} will be:
\begin{equation} \label{eqn:formalized-one-step-branch-lemma}
\forall X \exists Y \ ( ( \xi(X) \rightarrow ( \xi(Y) \land \chi(X,Y) ) ) \land ( \neg \xi(X) \rightarrow Y=Z_{\der}))
\end{equation}
To see that this can be proven within $\PSCA$, first note that for $X$ with $\neg \xi (X)$ we just let $Y=Z_{\der} $, so the hard case is when $ \xi (X)$ holds. The proof of this case mimics directly the meta-level proof of \autoref{lemma:one-step-branch-lemma}. This is where we need \autoref{prop:formalized-sat-respects-negation}, since we are dealing with unsatisfied formulas, and also the inductive characterization of satisfaction \autoref{theorem:formalised-satisfaction}. We also need \autoref{prop:formalized-approximant-recursion} for the $\nu^{\kappa}$-rule, \autoref{prop:formalized-knaster-tarski} for the $\mu$-rule and \autoref{prop:props-ordinal-comparison} for the reset rule.  

We will only go through a single case for illustration. Assume $\xi (X)$ and that $ \constr : \Gamma$ is the conclusion of the $\nu^{\kappa_j}$-rule. Let $t \in \nu^{\kappa_j} X_i x_i. \phi_i$ be the principal formula in the conclusion and $\phi_i( \nu^{\kappa_{j'}} X_i x_i. \phi_i, t) $ the active  formula of the premise. Then $j' \neq j$ and $\kappa_{j'}$ \emph{does not occur} in $t \in \nu^{\kappa_j} X_i x_i. \phi_i$. This means we have 

\begin{equation}\label{eqn:lukas-star}
   \sat{\der} (\rho, \code {\phi_i(x_i,X_i)} , \vec A, \alpha_0 , \hdots \alpha_{j'} , \hdots , \alpha_k) \quad  \leftrightarrow\quad  \sat{\der} (\rho, \code {\phi_i(x_i,X_i)} , \vec A, \alpha_0 , \hdots \beta , \hdots , \alpha_k) 
\end{equation}
 for any $\beta$.

Now write 
\[B(\beta) := \nu^{\beta} Y_i y_i . \sat{\der} (\rho\{ x_i \mapsto y_i \}, \code {\phi_i(x_i,X_i)} , A_0 , \hdots Y_i , \hdots A_k , \vec \alpha)  \]

Then we have the following equivalences:

\[
\def\arraystretch{1.2}
\begin{array}{cl}
   \sat{\der} (\rho, \code { s \in X_i } , A_0 , \hdots B(\beta) , \hdots A_k , \vec \alpha)\\
   \leftrightarrow & \text{Thm.}~\ref{theorem:formalised-satisfaction}\\
    \rho(s) \in \nu^{\beta} Y_i y_i . \sat{\der} (\rho\{ x_i \mapsto y_i \}, \code {\phi_i(x_i,X_i)} , A_0 , \hdots Y_i , \hdots A_k , \vec \alpha)  \\
    \leftrightarrow &\eqref{eqn:lukas-star}\\
     \rho(s) \in \nu^{\beta} Y_i y_i . \sat{\der} (\rho\{ x_i \mapsto y_i \}, \code {\phi_i(x_i,X_i)} , A_0 , \hdots Y_i , \hdots A_k , \alpha_0 , \hdots \beta , \hdots , \alpha_k )  \\
     \leftrightarrow  & \text{Thm.}~\ref{theorem:formalised-satisfaction}\\
       \sat{\der} (\rho, \code { s \in  \nu^{\kappa_{j'}} X_i x_i. \phi_i(x_i,X_i)} , \vec A , \alpha_0 , \hdots \beta , \hdots , \alpha_k) 
\end{array}
\]

By structural induction on $\phi_i$ we then get:

\begin{equation}\label{eqn:lukas-second-star}
  \def\arraystretch{1.2}
\begin{array}{c}
     \sat{\der} (\rho, \code { \phi_i( x_i  , X_i)  } , A_0 , \hdots B(\beta) , \hdots A_k , \vec \alpha)\\
     \leftrightarrow\\
     \sat{\der} (\rho, \code { \phi_i ( x_i , \nu^{\kappa_{j'}} X_i x_i. \phi_i )} , \vec A , \alpha_0 , \hdots \beta , \hdots , \alpha_k ) 
\end{array}  
\end{equation}

Finally we get:

\def\arraystretch{1.5}
\begin{tabular}{lll} 

 $(i)$ & $\neg \sat{\der} (\rho, \code{t \in \nu^{\kappa_j} X_i x_i. \phi_i(x_i, X_i) }, \vec A , \vec \alpha )$ & assumption \\
 
$(ii)$ & $\rho(t) \notin  \nu^{\alpha_j} Y_i y_i . \sat{\der} (\rho\{ x_i \mapsto y_i \}, \code {\phi_i(x_i, X_i)} , A_0 , \hdots Y_i , \hdots A_k , \vec \alpha) )$ & by Thm. \ref{theorem:formalised-satisfaction} \\

$(iii)$ & $ \exists \beta \prec \alpha_j \ 
\neg \sat{\der} (\rho\{ x_i \mapsto \rho(t) \}, \code {\phi_i(x_i,X_i)} , A_0 , \hdots B(\beta) , \hdots A_k , \vec \alpha) $ &  Prop. \ref{prop:formalized-approximant-recursion}  \\

$(iv)$ & $ \neg \sat{\der} (\rho\{ x_i \mapsto \rho(t) \}, \code {\phi_i(x_i,X_i)} , A_0 , \hdots B(\beta) , \hdots A_k , \vec \alpha) $ & instantiation\\

$(v)$ & $  \neg \sat{\der} (\rho, \code {\phi_i(t , X_i)} , A_0 , \hdots B(\beta) , \hdots A_k , \vec \alpha)$ &  induct. on $\phi_i$\\

$(vi)$ & $ \neg \sat{\der} (\rho, \code {\phi_i(t , \nu^{\kappa_j'} Y_i y_i. \phi_i)} , \vec A ,  \alpha_0 , \hdots \beta , \hdots , \alpha_k )$  &  $\eqref{eqn:lukas-second-star}$ \\

\end{tabular}

Thus, $\xi(Y)$ and $\chi(X,Y)$ for $Y= \langle \rho, \code{ \constr': \Gamma' } , \alpha_0 , \hdots \beta , \hdots , \alpha_k \rangle $ where $\constr': \Gamma'$ is the premise of $\constr: \Gamma$ in $\der$.

After proving \eqref{eqn:formalized-one-step-branch-lemma} we want to use $\Sij 1 2 $-dependent choice to get a suitable branch. Since $\xi(X)$ and $\chi(X,Y)$ are $\Dij 1 2$ we get by \autoref{thm:arithm-comb-of-delta12} that  \[( \xi(X) \rightarrow ( \xi(Y) \land \chi(X,Y) ) ) \land ( \neg \xi(X) \rightarrow Y=Z_{\der})\] is $\Dij 1 2$. Thus, we can apply $\Sij 1 2 $-dependent choice (see~\eqref{eqn:dependent-choice}) to conclude that there exists $F : \Nat \rightarrow SET$ s.t. 
\[\forall n  \ ( ( \xi(F(n)) \rightarrow ( \xi(F(n+1)) \land \chi(F(n),F(n+1)) ) ) \land ( \neg \xi(F(n)) \rightarrow F(n+1)=Z_{\der})) \]

Note that if $\neg \xi ( F(0) )$ then $F(1) = Z_{\der} $ and so $ \xi ( F(1) )$, which means by an easy induction that $\forall n \geq 1 \ \xi ( F(n) )$. By comprehension we define $B(n) := F(n+1)$ and so $B$ is a  branch witnessing falsity satisfying $\forall n \   \chi ( B(n) , B(n+1) ) $.   

Now we apply \autoref{lemma:reset-condition-provable} to get a stable reset  variable $\kappa_j$ such that, for some $m$ and all $n \geq m$, $\kappa_j$ occurs in $\constr$, where $B(n)= \langle \rho,  \code{  \constr : \Gamma} , \vec \alpha \rangle $. 

Finally, using this stable reset variable we want to show the existence of an infinitely descending ordinal sequence which would contradict \autoref{prop:props-ordinal-comparison}(4). Using comprehension we get $G: \Nat \rightarrow WO $ s.t. 
$$G: n \mapsto \alpha \ \leftrightarrow  \ B(n+m) = \langle \rho, \code{ \constr : \Gamma }, \alpha_0, \hdots \alpha_{j-1}, \alpha, \hdots \alpha_k \rangle  $$.

Since $\forall n  \  \chi ( B(n) , B(n+1) ) $ we see that $G$ contradicts $WF(\prec )$ from \autoref{prop:props-ordinal-comparison}(4). 

\end{proof}

Finally, we can put everything together to prove \autoref{thm:main-result}. 

\begin{proof}[Proof of Theorem 1]

 Let $\phi$ be arithmetical. If $\muPA \proves \phi$, then $\cmuPA \proves \phi$ by \autoref{id-to-cid}.

Conversely, suppose that $\cmuPA \proves \phi$. By \autoref{thm:translation}, there is a $\cmuPAord$ proof $\der$ with conclusion $\emptyset : \phi$. Hence, by \autoref{thm:formalisation-of-soundness},
\[
\PSCA \proves
\forall \vec \alpha \in WO \ \forall \rho,\vec A \
\sat{\der}(\rho,\code{\emptyset : \phi},\vec A,\vec\alpha).
\]
Since the constraint is empty and the sequent consists only of the arithmetical sentence $\phi$, we have
\[
\PSCA \proves \phi .
\]
by \autoref{reflection}. Finally, by M\"{o}llerfeld''s conservativity theorem, \autoref{thm:moellerfeld-psca-cons-over-muPA}, $\PSCA$ is arithmetically conservative over $\muPA$. Therefore $\muPA \proves \phi$, as required.
    
\end{proof}

\section{Conclusions and future works}

We studied the cyclic proof theory of arithmetic with generalised inductive definitions. Specifically,  we extended Simpson's equivalence result between Peano arithmetic and its cyclic formulation~\cite{Simpson17:cyc-arith}, further generalising Das and Melgaard's work on (finitely iterated) arithmetical inductive definitions~\cite{Das21:CT-fscd}. Along the way, we also established an equivalence between two alternative validity conditions for cyclic proofs in our setting: the commonly adopted condition based on the notion of progressing trace and so-called reset condition for annotated proof systems (e.g.,~\cite{sprengerdam03:journal,Graham-cyclic-first-order-mu-calculus}). 

Our results easily apply to the intuitionistic context, essentially using negative translations studied in~\cite{Tupailo04doubleneg,Curzi023}. Because of the absence of classical logic dualities, we can identify weaker versions of  generalised inductive definitions, such as \emph{strictly positive} ones (defined by fixed points where $\mu$ does not bind set variables to the left of an implication). Such subsystems are of great importance for implementation perspectives, as proof assistants (such as Coq or Agda) typically enforce strict positivity conditions on (co)inductive types for consistency reasons. Despite being apparently more restrictive,  strictly positive fixed points are commonly believed to have the same proof theoretic strength of the (hereditarily) positive ones (see~\cite{feferman1981iterated}). These considerations deserve further investigations.

\subsection*{Acknowledgements} We would like to thank Anupam Das for his continuous support, and Graham Leigh for his valuable suggestions during the development of this paper. 

This work was supported by a UKRI Future Leaders Fellowship, ``Structure vs Invariants
in Proofs'' (project reference MR/S035540/1),  by the Wallenberg Academy Fellowship Prolongation project ``Taming
Jörmungandr: The Logical Foundations of Circularity'' (project reference 251080003), and by the VR starting grant ``Proofs with Cycles in Computation'' (project reference 251088801).

\bibliographystyle{alphaurl}
\bibliography{main}

\newpage
\appendix

\section{Some results of $\muPA$}\label{app:derivations}

\begin{prop}\label{prop:axiomatic-implies-sequent}
 The axioms    ($\preaxiom$-$\mu$), ($\indaxiom$-$\mu$), ($\preaxiom$-$\nu$), and ($\indaxiom$-$\nu$) are derivable in the sequent calculus presentation of $\muPA$.
\end{prop}
\begin{proof}
    By the dualities for $\nu$ and $\mu$, it suffices to prove the last two axioms:
\begin{itemize}
    \item ($\preaxiom$-$\nu$):
    \[
\vlderivation{
\vlin{\forall}{}{\forall y(y \in \nu Xx\phi \to \phi(\nu Xx\phi, y))}
   {
   \vlin{\vee}{}{y \in \nu Xx\phi \to \phi(\nu Xx\phi, y)}
     {
     \vlid{=}{}{\neg y \in \nu Xx\phi,  \phi(\nu Xx\phi, y)}
         {
         \vlin{\mu}{}{y \in \mu Xx\neg \phi(\neg X),  \phi(\nu Xx\phi, y)}
          {
          \vlid{=}{}{\neg \phi(\neg \mu Xx\neg \phi(\neg X), y),   \phi(\nu Xx\phi, y)}{\vlin{\id}{}{\neg \phi (\nu Xx\phi),  \phi(\nu Xx\phi, y)}{\vlhy{}}}
          }
         }
     }
   }
}
    \]
    \item ($\indaxiom$-$\nu$):
    \[
\vlderivation{
\vlin{\vee}{}{\forall y(\psi(y)\to \phi(\psi, y))\to \forall y(\psi(y)\to y \in \nu Xx\phi)}
  {
 \vlin{\forall}{}{\exists y(\psi(y)\wedge \neg (\phi(\psi, y))),  \forall y(\psi(y)\to y \in \nu Xx\phi)}
    {
    \vlin{\vee}{}{\exists y(\psi(y)\wedge \neg (\phi(\psi, y))),  \psi(y)\to y \in \nu Xx\phi}
       {
       \vliin{\xind{\phi}}{}{\Gamma, y \in \nu Xx\phi}{
       \vlin{\id}{}{\Gamma, \psi(y)}{\vlhy{}}
       }
       {
       \vlin{\exists}{}{\Gamma, \neg \psi (z), \phi(\psi, z)}
         {
         % \vlin{\exists}{}{\Gamma, \neg \psi(z), \phi(\psi,z)} 
         %     {
             \vliin{\wedge}{}{\psi(z)\wedge \neg \phi(\psi, z), \neg \psi(y), \neg \psi(z), \phi(\psi, y)}
               {
               \vlin{\id}{}{\psi(z), \neg \psi(y), \neg \psi(z), \phi(\psi, y)}{\vlhy{}}
               }{
               \vlin{\id}{}{ \neg \phi(\psi, z), \neg \psi(y), \neg \psi(z), \phi(\psi, y)}{\vlhy{}}
               }
             }
         }
       % }
       }
    }
  }
}
    \]
    where $\Gamma=\exists y(\psi(y)\wedge \neg \phi(\psi, y)), \neg  \psi(y)$.
\end{itemize}
\end{proof}

 \begin{prop}\label{prop:funct-ind}
 If $\phi$ is a formula and $Y$ is a free set variable occurring positively in $\phi$ then $\neg \phi(\neg Y), \phi(Y)$ is derivable in $\muPA$ from $ z \not\in Y, z \in Y$.
 \end{prop}  
\begin{proof}
    The proof is by induction on the structure of $\phi$. We only consider the most relevant cases:
    \begin{itemize}
        \item If $\phi(Y) = t \in Y$  for some $t$:
        \[
\vlderivation{
\vlinf{[t/y]}{}{ t \not\in Y, t \in Y}{ z \not \in Y, z \in Y}
}
        \]
        \item If $\phi(Y)=t \in \nu Xx.\phi(X, Y)$ then by the induction hypothesis we have:
\[
\vlderivation{
\vliq{IH}{}{\neg \phi(X, y), \phi(X, y)}{\vlhy{z \not \in Y, z \in Y}}
}
\]
        
        \[
\vlderivation{
\vliin{\xind{\phi}}{}{ t \in \nu Xx \phi,   t \not \in \nu Xx \phi}{\vlin{\id}{}{  t\not \in \nu Xx \phi,   t \in \nu Xx \phi}{\vlhy{}}}
{
\vlid{=}{}{ y\not \in \nu Xx \phi, \phi(\nu Xx\phi, y)}
 {
 \vlin{\mu}{}{ y \in \mu Xx \neg \phi(\neg X), \phi(\nu Xx\phi, y)}
 {
 \vlid{=}{}{\neg \phi(\neg y \in \mu Xx \neg \phi(\neg X)), \phi(\nu Xx\phi, y)}{\vliq{IH}{}{\neg \phi(\nu Xx\phi, y), \phi(\nu Xx\phi, y)}{\vlhy{z \not \in Y, z \in Y}}
 }
 }
}
}
}
        \]
    \end{itemize}
\end{proof}

%%%%%%%%%%%%%%%%%%%%%%%%%%%%%%%%%%%%%%%%%%%%%%%%%%%%%%%%%%%%%%%%%% 
%%%%%%%%%%%%%%%%%%%%%%%%%%%%%%%%%%%%%%%%%%%%%%%%%%%%%%%%%%%%%%%%%% 
%CSL VERSION BEGINS  
%%%%%%%%%%%%%%%%%%%%%%%%%%%%%%%%%%%%%%%%%%%%%%%%%%%%%%%%%%%%%%%%%% 
%%%%%%%%%%%%%%%%%%%%%%%%%%%%%%%%%%%%%%%%%%%%%%%%%%%%%%%%%%%%%%%%%%  
% \begin{proposition}\label{prop:nu-derivable}
%     The rule $\nu$ is derivable in the sequent calculus presentation of $\muPA$.
% \end{proposition}
% \begin{proof}
% By using~\Cref{prop:funct-ind} we have:
% \end{proof}

% \begin{proposition}\label{subsuming-numerical-induction} The induction axiom $\PA _8$ (\Cref{fig:PA}) can be derived from the sequent calculus presentation of $\muPA$.
%     \end{proposition}
% \begin{proof}
%     By setting $\Gamma=\neg \phi(0), \neg \forall x(\phi(x)\to \phi(\succ x)), \phi(t)$ we have:
% \end{proof}
%%%%%%%%%%%%%%%%%%%%%%%%%%%%%%%%%%%%%%%%%%%%%%%%%%%%%%%%%%%%%%%%%% 
%%%%%%%%%%%%%%%%%%%%%%%%%%%%%%%%%%%%%%%%%%%%%%%%%%%%%%%%%%%%%%%%%% 
%CSL VERSION ENDS  
%%%%%%%%%%%%%%%%%%%%%%%%%%%%%%%%%%%%%%%%%%%%%%%%%%%%%%%%%%%%%%%%%% 
%%%%%%%%%%%%%%%%%%%%%%%%%%%%%%%%%%%%%%%%%%%%%%%%%%%%%%%%%%%%%%%%%%  

%    \begin{proposition}
% [Admissibility of substitution]
% \label{subst-admissibility}
% If there is a (non-wellfounded) $\muPAnoind$-proof of a sequent $\Gamma $, then there is one not using the substitution rule.
% \end{proposition}
% \begin{proof}
%     \todoinline{Prove it (I don't think we need it anyway)}
% \end{proof}
\section{Some results of $\cmuPA$}\label{app:cmupa}

\begin{prop}\label{prop:funct-cycl}
 If $\phi$ is a formula and $Y$ is a free set variable occurring positively in $\phi$ then $\neg \phi(\neg Y), \phi(Y)$ is derivable in $\cmuPA$ from $ z \not \in Y, z \in Y$.
 \end{prop}  
\begin{proof}
    The statement is proven by induction on the structure of $\phi$ as for~\autoref{prop:funct-ind}. The only critical case is when $\phi(Y)=t \in \nu Xx. \phi(X, Y)$:
\[
\vlderivation{
\vlin{\nu}{\bullet}{  t\not \in \nu Xx. \phi(X, Y),t  \in \nu Xx. \phi(X, Y)}
  {
  \vlid{=}{}{   t\not  \in \nu Xx. \phi(X, Y), \phi((\nu Xx. \phi(X, Y)),  Y, t)}{
  \vlin{\mu}{}{  t  \in \mu Xx. \neg \phi(\neg X, Y), \phi((\nu Xx. \phi(X, Y))}
   {
   \vliiq{IH}{}{  \neg \phi((\nu Xx \phi(X, Y)), Y, t), \phi((\nu Xx. \phi(X, Y))}
     {
     \vlin{}{\bullet}{  t\not  \in \nu Xx. \phi(X, Y), t  \in \nu Xx. \phi(X, Y)}{\vlhy{}}
     }{
     \vlhy{t \not \in Y, t \in Y}
     }
   }
  }
  }
}
\]
it is easy to check that the coderivation above is a (cyclic) progressing proof.
\end{proof}

\begin{prop}\label{prop:gen-id}
The identity rule can be derived in $\cmuPA$ from the following 
    \[
\vlinf{\id}{\phi \text{ atomic}}{\Gamma, \neg \phi, \phi}{}
    \]
\end{prop}
\begin{proof}
This is a straightforward induction on the structure of $\phi$. We only consider the case where $\phi= t \in \mu Xx \phi$:
\[
\vlderivation{
\vlid{=}{\bullet}{ t \not \in \mu Xx \phi, t \in \mu Xx \phi}
 {
 \vlin{\mu}{}{ t \in \nu Xx \neg \phi(\neg X), t \in \mu Xx \phi}
   {
   \vlin{\nu}{}{  t \in \nu Xx \neg \phi(\neg X), \phi(\mu Xx\phi, t)}
      {
      \vlid{=}{}{\neg \phi(\neg \nu Xx \neg \phi(\neg X), t), \phi(\mu Xx\phi, t)}
         {
         \vliq{IH}{}{\neg \phi(\mu Xx\phi, t), \phi(\mu Xx\phi, t)}{
         \vlin{=\text{-}\mathsf{sub}}{}{  t'\not  \in \mu Xx\phi, t' \in \mu Xx\phi}{
         \vlin{}{\bullet}{  t \not\in \mu Xx\phi, t \in \mu Xx\phi}{\vlhy{}}
         }
         }
         }
      }
   }
 }
}
\]
it is easy to check that the coderivation above is a (cyclic) progressing proof.
\end{proof}

\end{document}